\documentclass[11pt]{article}

\usepackage[top=2cm, bottom=2cm, left=2cm, right=2cm]{geometry} %
\usepackage{amsmath}
\usepackage{amsfonts}
\usepackage{amsthm}
\usepackage{complexity}
\usepackage{caption}
\usepackage{subcaption}
\usepackage{graphicx, color}
\usepackage{epsfig}
\usepackage{paralist}
\usepackage[ruled,lined,linesnumbered]{algorithm2e}
\usepackage{epstopdf}
\usepackage{authblk}
\usepackage{enumerate}
\usepackage{enumitem}

\newtheorem{Lemma}{Lemma}
\newtheorem{Definition}{Definition}

\begin{document}

\title{NED: An Inter-Graph \underline{N}ode Metric Based On \underline{E}dit \underline{D}istance}

\author{Haohan Zhu}
\author{Xianrui Meng}
\author{George Kollios}

\affil{Department of Computer Science, Boston University\authorcr \{ zhu, xmeng, gkollios \}@cs.bu.edu}
\date{}

\maketitle

\begin{abstract}
Node similarity is a fundamental problem in graph analytics. However, node similarity between nodes in different graphs (inter-graph nodes) has not received a lot of attention yet. The inter-graph node similarity is important in learning a new graph based on the knowledge of an existing graph (transfer learning on graphs) and has applications in biological, communication, and social networks. In this paper, we propose a novel distance function for measuring inter-graph \underline{n}ode similarity with \underline{e}dit \underline{d}istance, called \textbf{NED}. In NED, two  nodes are compared according to their local neighborhood structures which are represented as unordered $k$-adjacent trees, without relying on labels or other assumptions. Since the computation problem of tree edit distance on unordered trees is NP-Complete, we propose a modified tree edit distance, called \textbf{TED*}, for comparing neighborhood trees. TED* is a metric distance, as the original tree edit distance, but more importantly, TED* is polynomially computable. As a metric distance, NED admits efficient indexing, provides interpretable results, and shows to perform better  than existing approaches on a number of data analysis tasks, including graph de-anonymization. Finally, the efficiency and effectiveness of NED are empirically demonstrated  using real-world graphs.
\end{abstract}

\section{Introduction}
\label{sec: Introduction}

Node similarity is a fundamental problem in graph data analysis. Many applications use node similarity as an essential building block to develop more complex graph data mining algorithms. These applications include node classification, similarity retrieval, and graph matching.

In particular, node similarity measures between nodes in different graphs (inter-graph nodes) can have many important applications including transfer learning  across networks and graph de-anonymization~\cite{DBLP:conf/kdd/HendersonGLAETF11}. An example comes from biological networks. It has been recognized that the topological (neighborhood) structure of a node in a biological network (e.g., a PPI network) is related to the functional and biological properties of the node~\cite{DBLP:journals/bioinformatics/DavisYMSP15}. Furthermore, with the increasing production of new biological data and networks, there is an increasing need to find nodes in these new networks that have similar topological structures (via similarity search) with nodes in already analyzed and explored networks~\cite{DBLP:journals/bioinformatics/ClarkK14}. Notice that additional information on the nodes can be used to enhance the distance function that we compute using the network structure around the node.
Another application comes from communication networks.  Consider a set of IP communication graphs from different days or different networks have been collected  and only one of these networks has been analyzed~\cite{DBLP:conf/kdd/HendersonGLAETF11}.  For example, nodes in one network may have been classified into classes or roles based on their neighborhood.  The question is how to use this information to classify nodes from the other networks without the need to build new classifiers (e.g., across-network classification) ~\cite{DBLP:conf/kdd/HendersonGLAETF11}. Finally, another important application of inter-graph node similarity is to use it for de-anonymization. As an example, given an anonymous social network and a non-anonymized social graph in the same domain, we can compare pairwise nodes to de-anonymize or re-identify the nodes in the anonymous social network by using the structural information from the non-anonymized communication graph \cite{DBLP:conf/kdd/HendersonGLAETF11}.

In recent years, many similarity measures for nodes in graphs have been proposed but most of them work only for nodes in the same graphs (intra-graph.) Examples include SimRank \cite{DBLP:conf/kdd/JehW02}, SimRank variants \cite{DBLP:journals/pvldb/YuLZCP13,DBLP:journals/pvldb/AntonellisGC08,DBLP:conf/kdd/JinLH11}, random walks with restart \cite{DBLP:conf/icdm/TongFP06}, and set matching methods~\cite{doi1971,ASI1973,DBLP:conf/kdd/XuYFS07,DBLP:conf/sigmod/KhanLYGCT11,DBLP:journals/pvldb/KhanWAY13}. Unfortunately these methods cannot be used for inter-graph nodes.  
Existing methods that can be used for inter-graph similarity  \cite{DBLP:conf/pakdd/AkogluMF10,DBLP:conf/asunam/BerlingerioKEF13,DBLP:conf/kdd/HendersonGLAETF11,Blondel2004}, they all have certain problems. OddBall \cite{DBLP:conf/pakdd/AkogluMF10} and NetSimile \cite{DBLP:conf/asunam/BerlingerioKEF13} only consider the features in the ego-net (instant neighbors) which limits the neighborhood information. On the other hand, more advanced methods that consider larger neighborhood structures around nodes, like  ReFeX \cite{DBLP:conf/kdd/HendersonGLAETF11} and HITS-based similarity \cite{Blondel2004} are not metric distances and the distance values between different pairs are not comparable. Furthermore, the distance values are not easy to interpret. 


In this paper, we propose a novel distance function for measuring inter-graph \underline{n}ode similarity with \underline{e}dit \underline{d}istance, called \textbf{NED}. In our measure, two inter-graph nodes are compared according to their neighborhood topological structures which are represented as  unordered $k$-adjacent trees. In particular, the NED between a pair of inter-graph nodes is equal to a modified tree edit distance called \textbf{TED*} that we also propose in this paper. We introduce TED* because the computation of the original tree edit distance on  unordered $k$-adjacent trees belongs to $\NP$-Complete. TED* is not only polynomially computable, but it also preserves all metric properties as the original tree edit distance does. TED* is empirically demonstrated to be efficient and effective when comparing trees. Compared to existing inter-graph node similarity measures, NED is a  metric and therefore  can admit efficient indexing and provides results that are interpretable, since it is based on edit distance. Moreover, since we can parameterize the depth of the neighborhood structure around each node, NED provides better quality results in our experiments, including better precision on  graph de-anonymization.  We provide a detailed experimental evaluation than show the efficiency and effectiveness of NED  using a number of real-world graphs.

Overall, in this paper we make the following contributions:
\begin{enumerate}[label={[\arabic*]}]
\item We propose a novel distance function, NED, to measure the similarity between inter-graph nodes.
\item We propose a modified tree edit distance, TED*, on unordered trees that is both metric and polynomially computable.
\item By using TED*, NED is an interpretable and precise node similarity that can capture the neighborhood topological differences.
\item We show that NED  can admit efficient indexing for similarity retrieval.
\item We experimentally evaluate NED using real datasets and we show that it performs better than existing approaches on  a de-anonymization application.
\end{enumerate}

The rest of this paper is organized as follows. Section 2 introduces related work of node similarity measurements. Section 3 proposes the unordered $k$-adjacent tree, inter-graph node similarity with edit distance which is called as NED in this paper, and the NED in directed graphs. In Section 4, we introduce TED*, our modified tree edit distance, and its edit operations. Section 5 elaborates the detailed algorithms for computing TED*. Section 6 presents the correctness proof of the algorithm and Section 7 proves the metric properties of TED*. In Section 8, we illustrate the isomorphism computability issue when comparing node similarities. Section 9 provides the analysis of the complexities and Section introduces the analysis of the only parameter $k$ in NED and the monotonicity property in NED. Section 11 theoretically compares the TED* we propose in this paper with TED and GED from two aspects: edit operations and edit distances. In Section 12, we propose a weighted version of NED. Section 13 empirically verifies the effectiveness and efficiency of our TED* and NED. Finally, we conclude our paper in Section 14.

\section{Related Work}
\label{sec: Related Work}

One major type of node similarity measure is called link-based similarity or transitivity-based similarity and is designed to compare intra-graph nodes. SimRank \cite{DBLP:conf/kdd/JehW02} and a number of SimRank variants like SimRank* \cite{DBLP:journals/pvldb/YuLZCP13}, SimRank++ \cite{DBLP:journals/pvldb/AntonellisGC08}, RoleSim \cite{DBLP:conf/kdd/JinLH11}, just to name a few, are typical link-based similarities which have been studied extensively. Other link-based similarities include random walks with restart \cite{DBLP:conf/icdm/TongFP06} and path-based similarity \cite{DBLP:journals/pvldb/SunHYYW11}. A comparative study for link-based node similarities can be found in \cite{DBLP:journals/tkde/LiuHZLD13}. Unfortunately, those link-based node similarities are not suitable for comparing inter-graph nodes since these nodes are not connected and the distances will be always $0$. 

To compare inter-graph nodes, neighborhood-based similarities have been used. Some primitive methods directly compare the ego-nets (direct neighbors) of two nodes using Jaccard coefficient, S{\o}rensen{--}Dice coefficient, or Ochiai coefficient \cite{doi1971,ASI1973,DBLP:conf/kdd/XuYFS07}. Ness \cite{DBLP:conf/sigmod/KhanLYGCT11} and NeMa \cite{DBLP:journals/pvldb/KhanWAY13} expand on this idea and they use the structure of the $k$-hop neighborhood of each node. However, for all these methods, if two nodes do not share common neighbors (or neighbors with the same labels), the distance will always be $0$, even if the neighborhoods are isomorphic to each other.

An approach that can work for inter-graph nodes is to extract features from each node using the neighborhood structure and compare these features.   OddBall \cite{DBLP:conf/pakdd/AkogluMF10} and NetSimile \cite{DBLP:conf/asunam/BerlingerioKEF13} construct the feature vectors by using the ego-nets (direct neighbors) information such as the degree of the node, the number of edges in the ego-net and so on. ReFeX \cite{DBLP:conf/kdd/HendersonGLAETF11} is a framework to construct the structural features recursively. The main problem with this approach is that the choice of features is ad-hoc and the distance function is not easy to interpret. Furthermore, in many cases, the distance function may be zero even for nodes with different neighborhoods. Actually, for the more advanced method, ReFeX, the distance method is not even a metric distance.

Another method that has been used for comparing biological networks, such as protein-protein interaction networks (PPI) and metabolic networks, is to extract a feature vector using graphlets~\cite{DBLP:journals/bioinformatics/Malod-DogninP15,DBLP:journals/bioinformatics/DavisYMSP15}. Graphlets are small connected non-isomorphic induced subgraphs of a large network~\cite{DBLP:journals/bioinformatics/PrzuljCJ04} and generalize the notion of the degree of a node. However, they are also limited to the small neighborhood around each node and as the size of the neighborhood increases the accuracy of this method decreases.

Another node similarity for inter-graph nodes based only on the network structure is proposed by Blondel et al. \cite{Blondel2004} which is called HITS-based similarity. In HITS-based similarity, all pairs of nodes between two graphs are virtually connected. The similarity between a pair of inter-graph nodes is calculated using the following similarity matrix:
\begin{equation*}
S_{k+1} = BS_kA^T + B^TS_kA
\end{equation*}
where $A$ and $B$ are the adjacency matrices of the two graphs and $S_k$ is the similarity matrix in the $k$ iteration.

Both HITS-based and Feature-based similarities are capable to compare inter-graph nodes without any additional assumption. However, HITS-based similarity is neither metric nor efficient. On the other hand, Feature-based similarities use ad-hoc statistical information which cannot distinguish minor topological differences. This means that Feature-based similarities may treat two nodes as equivalent even though they have different neighborhood structures.

\section{NED: Inter-Graph Node Similarity with Edit Distance}
\label{sec: NED: Inter-Graph Node Similarity with Edit Distance}

\subsection{K-Adjacent Tree}
\label{sec: K-Adjacent Tree}

\begin{figure}
\centering
\includegraphics[scale=0.7]{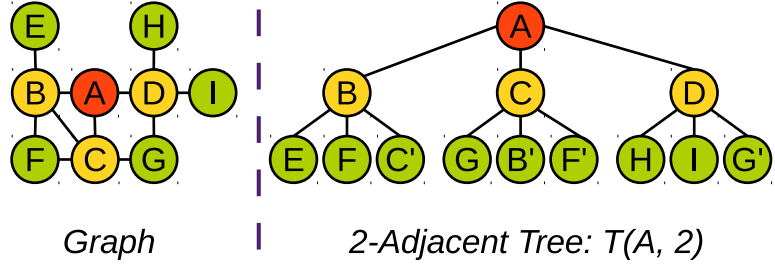}
\caption{K-Adjacent Tree}
\label{fig:K-Adjacent_Tree}
\end{figure}

First of all, we introduce the unlabeled unordered k-adjacent tree that we use to represent the neighborhood topological structure of each node. The $k$-adjacent tree was firstly proposed by Wang et al. \cite{DBLP:journals/tkde/WangWYY12} and for the completeness, we include the definition here:

\begin{Definition}
The adjacent tree $T(v)$ of a vertex $v$ in graph $G(V,E)$ is the breadth-first search tree starting from vertex $v$. The $k$-adjacent tree $T(v,k)$ of a vertex $v$ in graph $G(V,E)$ is the top $k$-level subtree of $T(v)$.
\label{Definition:K-Adjacent Tree}
\end{Definition}

The difference between the $k$-adjacent tree in \cite{DBLP:journals/tkde/WangWYY12} and the unlabeled unordered $k$-adjacent tree used in this paper is that the children of each node are not sorted based on their labels. Thus, the $k$-adjacent tree used in this paper is an unordered unlabeled tree structure.

An example of a $k$-adjacent tree is illustrated in Figure \ref{fig:K-Adjacent_Tree}. For a given node in a graph, its $k$-adjacent tree can be retrieved deterministically by using breadth first search. In this paper we use this tree to represent the topological neighborhood of a node that reflects its ``signature''  or ``behavior'' inside the graph.

In the following paper, we consider undirected graphs for simplicity. However, the $k$-adjacent tree can also be extracted from directed graphs. Furthermore, our distance metric can also be applied to directed graphs as we discuss in Section \ref{sec: NED in Directed Graphs}. Later in Section \ref{sec: Isomorphism Computability},  we explain why we chose to use a tree rather than the actual graph structure around a node as the node signature.

\subsection{NED}
\label{sec:NED}

Here, we introduce NED, the inter-graph node similarity with edit distance. Let $u$ and $v$ be two nodes from two graphs $G_u$ and $G_v$ respectively. For a given parameter $k$, two $k$-adjacent trees $T(u,k)$ and $T(v,k)$ of nodes $u$ and $v$ can be generated separately. Then, by applying the modified tree edit distance TED* on the pair of two $k$-adjacent trees, the similarity between the pair of  nodes is defined as the similarity between the pair of the two $k$-adjacent trees.

Denote $\delta^k$ as the distance function NED between two nodes with parameter $k$ and denote $\delta_T$ as the modified tree edit distance TED* between two trees. Then we have, for a parameter $k$,
\begin{equation}
\delta^k(u, v) = \delta_T(T(u,k),T(v,k))
\label{Equation:NED Equation}
\end{equation}

Notice that, the modified tree edit distance TED* is a generic distance for tree structures. In the following sections, we present the definition and algorithms for computing the proposed modified tree edit distance.

\subsection{NED in Directed Graphs}
\label{sec: NED in Directed Graphs}

In this paper, we discuss the inter-graph node similarity for undirected graphs. The $k$-adjacent tree is also defined in undirected graphs. However it is possible to extend the $k$-adjacent tree and inter-graph node similarity from undirected graphs to directed graphs. For directed graphs, there are two types of $k$-adjacent trees: incoming $k$-adjacent tree and outgoing $k$-adjacent tree. The definition of incoming $k$-adjacent tree is as follows:

\begin{Definition}
The incoming adjacent tree $T_I(v)$ of a vertex $v$ in graph $G(V,E)$ is a breadth-first search tree of vertex $v$ with incoming edges only. The incoming $k$-adjacent tree $T_I(v,k)$ of a vertex $v$ in graph $G(V,E)$ is the top $k$-level subtree of $T_I(v)$.
\label{Definition:K-Adjacent Tree in Directed Graphs}
\end{Definition}

Similarly, the outgoing adjacent tree $T_O(v)$ includes outgoing edges only.  For a given node in a directed graph, both incoming $k$-adjacent tree and outgoing $k$-adjacent tree can be deterministically extracted by using breadth-first search on incoming edges or outgoing edges only.

Based on Definition \ref{Definition:K-Adjacent Tree in Directed Graphs}, for a node $v$, there are two $k$-adjacent trees: $T_I(v)$ and $T_O(v)$.  Let $u$ be a node in the directed graph $G_u$ and $v$ be a node in the directed graph $G_v$. Then the distance function NED $\delta^k_D$ in directed graphs can be defined as:
\begin{equation}
\delta^k_D(u,v) = \delta_T(T_I(u),T_I(u)) + \delta_T(T_O(u),T_O(v))
\end{equation}

Notice that, since TED* is proved to be a metric, not only the NED defined in undirected graphs is a node metric, but the NED defined in directed graphs is also a node metric. The identity, non-negativity, symmetry and triangular inequality all preserve according to the above definition.

\section{TED*: Modified Tree Edit Distance}
\label{sec: TED*: Modified Tree Edit Distance}

\begin{figure*}
\centering
\includegraphics[scale=0.5]{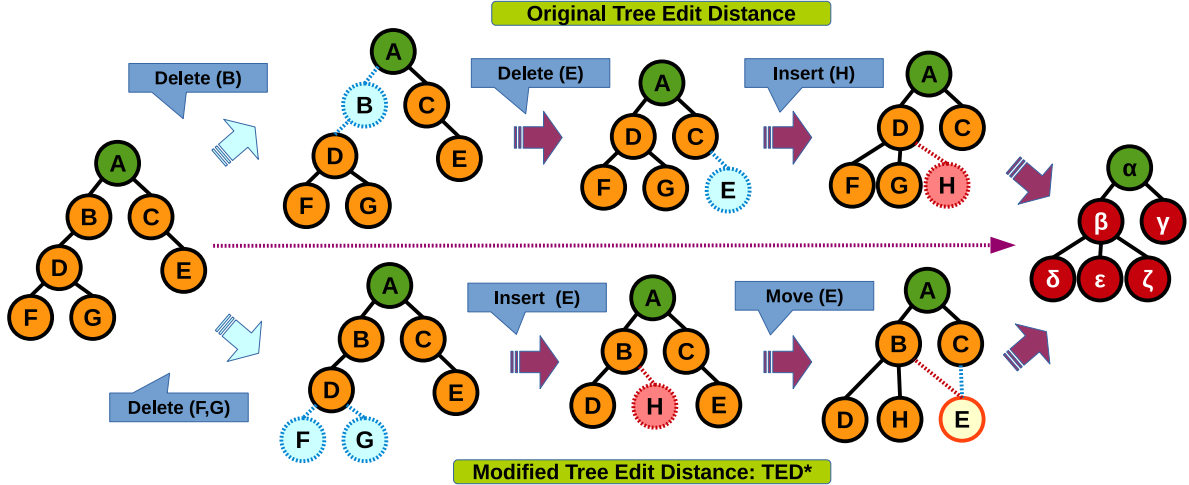}
\caption{TED* vs Tree Edit Distance}
\label{fig:TED*_TED}
\end{figure*}

The tree edit distance \cite{DBLP:journals/jacm/Tai79} is a well-defined and popular metric for tree structures. For a given pair of trees, the tree edit distance is the minimal number of edit operations which convert one tree into the other. The edit operations in the original tree edit distance include: 1) Insert a node; 2) Delete a node; and 3) Rename a node. For unlabeled trees, there is no rename operation.

Although the ordered tree edit distance can be calculated in $O(n^3)$ \cite{DBLP:journals/pvldb/PawlikA11}, the computation of unordered tree edit distance has been proved to belong to $\NP$-Complete \cite{DBLP:journals/ipl/ZhangSS92}, and even $\MaxSNP$-Hard \cite{DBLP:journals/ipl/ZhangJ94}. Therefore, we propose a novel modified tree edit distance called TED* which still satisfies the metric properties but it is also polynomially computable.

The edit operations in TED* are different from the edit operations in the original tree edit distance. In particular, in TED*, we do not allow any operation that can change the depth of any existing node in the trees. The reason is that, we view the depth of a neighbor node, which represents the closeness between the neighbor node and the root node, as an important property of this node in the neighborhood topology. Therefore, two nodes with different depths should not be matched. The definition of TED* is below:

where each edit operation $e_i$ belongs to the set of edit operations defined in Section \ref{sec: Edit Operations in TED*}. $|E|$ is the number of edit operations in $E$ and $\delta_T$ is the TED* distance proposed in this paper.

\subsection{Edit Operations in TED*}
\label{sec: Edit Operations in TED*}

In the original tree edit distance, when inserting a node between an existing node $n$ and its parent, it increases the depth of node $n$ and also increases the depths of all the descendants of node $n$. Similarly, when deleting a node which has descendants, it decreases the depths of all the descendants. Since in TED* we do not want to change the depth of any node, we should not allow these operations. Therefore, we need another set of edit operations as follows:

\begin{enumerate}[label=\Roman*:]  
\item Insert a leaf node
\item Delete a leaf node
\item Move a node at the same level
\end{enumerate}

To clarify, ``Move a node at the same level'' means changing an existing node's parent to another. The new parent node should be in the same level as the previous parent node. The above $3$ modified edit operations do not change the depth of any existing node. Also after any edit operation, the tree structure is preserved.

Figure \ref{fig:TED*_TED} shows an example of the difference between the traditional tree edit distance and our modified tree edit distance. When converting the tree $T_\alpha$ to the tree $T_\beta$,  the traditional tree edit distance requires $3$ edit operations: delete node $B$, delete node $E$ and insert node $H$. TED* requires $4$ edit operations: delete node $F$, delete node $G$, insert node $H$ and move node $E$.

Notice that, for the same pair of trees, TED* may be smaller or larger than the original tree edit distance. In Section \ref{sec:TED*, TED and GED}, we analyze the differences among TED*, the original tree edit distance and the graph edit distance in more details, where the TED* can be used to provide an upper-bound for the graph edit distance on trees.

In the following of this paper, number of edit operations is considered in TED* which means each edit operation in TED* has a unit cost. However, it is easy to extend TED* to a weighted version. In Section \ref{sec:Weighted TED*}, we introduce the weighted TED*. The weighted TED* can be proven to be a metric too and moreover, the weighted TED* can provide an upper-bound for the original tree edit distance.

\begin{Definition}
Given two trees $T_1$ and $T_2$, a series of edit operations $E=\{e_1, ...e_n \}$ is valid denoted as $E_v$, if $T_1$ can be converted into an isomorphic tree of $T_2$ by applying the edit operations in $E$. Then $\delta_T(T_1,T_2)$ $=$ $\min |E|$, $\forall E_v$.
\label{Definition:TED*}
\end{Definition}

\section{TED* Computation}
\label{sec: TED* Computation}

In this section, we introduce the algorithm to compute TED* between a pair of $k$-adjacent trees. It is easy to extend TED* to compare two generic unordered  trees. Before illustrating the algorithm, we introduce some definitions.

\begin{Definition}
Let $L_i(u)$ be the $i$-th level of the $k$-adjacent tree $T(u,k)$, where $L_i(u)$ $=$ $\{n|n \in T(u,k), d(n, u) = i \}$ and $d(n,u)$ is the depth of node $n$ in $T(u,k)$.
\label{Definition: Level}
\end{Definition}

In Definition \ref{Definition: Level}, the $i$-th level $L_i(u)$ includes the nodes with depths of $i$ in the $k$-adjacent tree $T(u,k)$. Similarly in $k$-adjacent tree $T(v,k)$, there exists the $i$-th level $L_i(v)$. The algorithm compares two $k$-adjacent trees $T(u,k)$ and $T(v,k)$ bottom-up and level by level. First, the algorithm compares and matches the two bottom levels $L_k(u)$ and $L_k(v)$. Then the next levels $L_{k-1}(u)$ and $L_{k-1}(v)$ are compared and matched and we continue like that until the root of the trees.

In our algorithm, we compare and match two levels based on the \textit{canonization labels} of nodes in the corresponding levels. The canonization label is defined as follows:

\begin{Definition}
Let $C(n)$ be the canonization label of node $n$, $C(n)$ $\in$ $\mathbb Z_{\geq 0}$. The canonization label $C(n)$ is assigned based on the subtree of node $n$. Two nodes $u$ and $v$ have the same canonization labels $C(u)$ $=$ $C(v)$, if and only if the two subtrees of nodes $u$ and $v$ are isomorphic.
\label{Definition:Canonization Label}
\end{Definition}

Although we use non-negative integer numbers to represent canonization labels in this paper, actually any set of symbols can be used for that. We use $x$ $\sqsubset$ $y$ to denote that node $x$ is a child of node $y$.

To compute TED* distance, we use two types of costs: padding cost and matching cost. Since the algorithm runs level by level,  for each level $i$, there exists a local padding cost $P_i$ and a local matching cost $M_i$.

The TED* distance represents the minimal number of modified edit operations needed to convert $T(u,k)$ to $T(v,k)$. Therefore, for each level $i$, after the comparing and matching, there exists a bijective mapping function $f_i$ from the nodes in $L_i(u)$ to the nodes in $L_i(v)$. Thus, $\forall$ $x$ $\in$ $L_i(u)$, $f_i(x)$ $=$ $y$ $\in$ $L_i(v)$ and $\forall$ $y$ $\in$ $L_i(v)$, $f_i^{-1}(y)$ $=$ $x$ $\in$ $L_i(u)$.

\begin{table}
\centering
\begin{tabular}{|c|l|} \hline
$T(u, k)$ & $k$-adjacent tree of node $u$\\ \hline
$L_i(u)$ & $i$th-level of $k$-adjacent tree of node $u$\\ \hline
$C(n)$ & Canonization label of node $n$\\ \hline
$x \sqsubset y$ & Node $x$ is a child of node $y$\\ \hline
$P_i$ & Padding cost for the level $i$ \\ \hline
$M_i$ & Matching cost for the level $i$\\ \hline
$G^2_i$ & Complete bipartite graph in the level $i$\\ \hline
$w(x,y)$ & Edge weight in $G^2_i$ between $x$ and $y$\\ \hline
$m(G^2_i)$ & Minimal cost for $G^2_i$ matching\\ \hline
$f_i: f_i(x) = y$ & Node mapping function for $G^2_i$ matching\\ \hline
\end{tabular}
\caption{Notation Summarization for TED* Algorithm}
\label{table: Notation Summarization for TED* Algorithm}
\end{table}

The notations used in the algorithm are listed in Table \ref{table: Notation Summarization for TED* Algorithm}.

\subsection{Algorithmic Overview}
\label{Algorithm Overview}

In this section, we present an overview of  the TED* computation in Algorithm \ref{alg:TED*_algorithm}. The input of the algorithm are two $k$-adjacent trees.

The TED* algorithm is executed bottom-up and level by level. For each level, there exists a padding cost and a matching cost. The TED* distance is the summation of padding  and matching costs from all levels.

Actually, as we explain in the following sections, there exists one-to-one mapping from padding and matching in the algorithm to the edit operations defined in Section \ref{sec: Edit Operations in TED*}. The padding cost represents the number of edit operations of 1) Inserting a leaf node and 2) Deleting a leaf node, and the matching cost is number of edit operations of 3) Moving a node at the same level.

To compute the padding  and matching costs, we use $6$ steps in each level: node padding (line 2-6 in Algorithm \ref{alg:TED*_algorithm}), node canonization (line 7-8), complete weighted bipartite graph construction (line 9-13),  weighted bipartite graph matching (line 14), matching cost calculation (line 15) and node re-canonization (line 16-19). Next, we describe those $6$ steps in details.

\subsection{Node Padding}
\label{Node Padding}

Node padding includes two parts:  cost calculation and node padding. The padding cost is the size difference between two corresponding levels. Let $L_i(u)$ and $L_i(v)$ be the corresponding levels. The difference between the number of nodes in $L_i(u)$ and the number of nodes in $L_i(v)$ is the padding cost:
\begin{equation}
P_i = \big||L_i(u)| - |L_i(v)|\big|
\label{Equation:Padding Cost}
\end{equation}

Actually the padding cost represents the number of edit operations of inserting a leaf node or deleting a leaf node. Assume two levels $L_i(u)$ and $L_i(v)$ of two $k$-adjacent trees $T(u, k)$ and $T(v, k)$ that have different number of nodes. Without loss of generality, let's suppose that $|L_i(u)|$ $<$ $|L_i(v)|$. Then when transforming the tree $T(u, k)$ to the tree $T(v, k)$, there must be several ``inserting a leaf node'' edit operations conducted on level $L_i(u)$. Symmetrically, if transforming the tree $T(v, k)$ to the tree $T(u, k)$, there must be several ``deleting a leaf node'' edit operations conducted on level $L_i(v)$. There is no other edit operation that can change the number of nodes in level $L_i(u)$ and level $L_i(v)$. In the node padding step, we always pad leaf nodes to the level which has less nodes.

\begin{algorithm}
\KwIn{Tree $T(u, k)$ and Tree $T(v, k)$ }
\KwOut{$\delta_T(T(u, k), T(v, k))$}
\For{$i \leftarrow k$ \KwTo $1$}{
   Calculate padding cost: $P_i$ $=$ $\big||L_i(u)|$ - $|L_i(v)|\big|$\;
   \uIf{$|L_i(u)|$ $<$ $|L_i(v)|$}{Pad $P_i$ nodes to $L_i(u)$\;}
   \ElseIf{$|L_i(u)|$ $>$ $|L_i(v)|$}{Pad $P_i$ nodes to $L_i(v)$\;}
   \ForEach{$n$ $\in$ $L_i(u)$ $\cup$ $L_i(v)$}{
       Get node canonization: $C(n)$\; 
   }
   \ForEach{$(x,y)$, where $x$ $\in$ $L_i(u)$ $\&$ $y$ $\in$ $L_i(v)$}{
       Get collection $S(x)$ $=$ $($ $C(x')$ $|$ $\forall x' \sqsubset x$ $)$\;
       Get collection $S(y)$ $=$ $($ $C(y')$ $|$ $\forall y' \sqsubset y$ $)$\;
       $w(x,y)$ $=$ $|S(x) \setminus S(y)|$ $+$ $|S(y) \setminus S(x)|$\; 
   }
   Construct bipartite graph $G^2_i$ with weights $w(x,y)$\; 
   Get cost $m(G^2_i)$ for minimal matching of $G^2_i$\;
   Calculate matching cost $M_i$ $=$ $(m(G^2_i)$ $-$ $P_{i+1})$ $/$ $2$\;
   \eIf{$|L_i(u)|$ $<$ $|L_i(v)|$}
   {\lForEach{$x$ $\in$ $L_i(u)$}{$C(x)$ $=$ $C(f_i(x))$}}
   {\lForEach{$y$ $\in$ $L_i(v)$}{$C(y)$ $=$ $C(f_i^{-1}(y))$}}
}
Return $\delta_T(T(u, k), T(v, k))$ $=$ $\sum_{i=1}^{k}(P_i + M_i)$\;
\caption{Algorithm for TED* Computation}
\label{alg:TED*_algorithm}
\end{algorithm}

\subsection{Node Canonization}
\label{Node Canonization}

\begin{figure}
\centering
\includegraphics[scale=0.7]{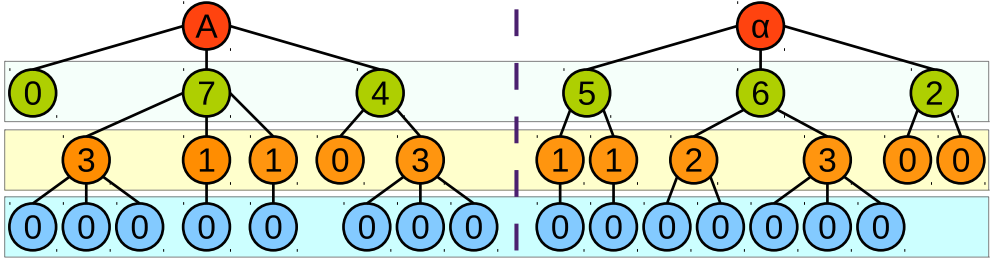}
\caption{Node Canonization}
\label{fig:Canonization}
\end{figure}

After node padding, we assign the canonization labels to all the nodes in the corresponding levels $L_i(u)$ and $L_i(v)$. Namely, $\forall$ $n$ $\in$ $L_i(u)$ $\cup$ $L_i(v)$, we assign the canonization label $C(n)$ to node $n$.

Based on Definition \ref{Definition:Canonization Label}, two nodes $x$ and $y$ have the same canonization labels $C(x)$ $=$ $C(y)$, if and only if the two subtrees of nodes $x$ and $y$ are isomorphic. However, we do not need to check the full subtrees of two nodes to decide whether those two nodes should be assigned to the same canonization label or not. We can use the children's canonization labels to decide whether two nodes have the same subtrees or not. Let $S(x)$ be the collection of the canonization labels of all the children of $x$, i.e. 
\begin{Definition}
$S(x) = (C(x'_1)...C(x'_{|x|}))$, where $x'_i \sqsubset x$ for $1\le i \le |x|$ and $|x|$ is the total number of node $x$'s children.
\end{Definition}

The collection of the canonization labels may maintain duplicate labels, since two children may have the same canonization labels. Also the canonization labels in a collection can be lexicographically ordered. Therefore, we have the following Lemma \ref{Lemma: Canonization Checking}.

\begin{Lemma}
$C(x)$ $=$ $C(y)$ iff $S(x)$ $\equiv$ $S(y)$.
\label{Lemma: Canonization Checking}
\end{Lemma}

Note that the equivalence $\equiv$ denotes that two collections $S(x)$ and $S(y)$ contain exactly the same elements. For example, $S(x) = (0,0,1)$ and $S(y)=(0,1,0)$ are equivalent. The proof is straightforward, since if the subtrees of two nodes are isomorphic, two nodes must have the same number of children and the corresponding children must have  isomorphic subtrees.

\begin{algorithm}
\KwIn{Two levels $L_i(u)$ and $L_i(v)$}
\KwOut{$C(n)$, $\forall$ $n$ $\in$ $L_i(u)$ $\cup$ $L_i(v)$}
Queue $q$ is lexicographically ordered\;
\ForEach{$n$ $\in$ $L_i(u)$ $\cup$ $L_i(v)$}{
    Get collection $S(n)$ $=$ $($ $C(n')$ $|$ $\forall n' \sqsubset n$ $)$\;
    $q$ $\leftarrow$ $S(n)$
}
Pop the first element in $q$ as $q_0 = S(x)$\;
$C(x)$ $=$ $0$\;
\For{$i \leftarrow 1$ \KwTo $|L_i(u)|+|L_i(v)|-1$}{
	\eIf{$q_i = S(y)$ $\equiv$ $q_{i-1} = S(x)$}{
		$C(y)$ $=$ $C(x)$
	}{
		$C(y)$ $=$ $C(x)$ $+$ $1$
	}
}
\caption{Node Canonization}
\label{alg:Canonization}
\end{algorithm}

To avoid checking all pairs of nodes when assigning canonization labels, we lexicographically order the collections of children's canonization labels. First, the collections are ordered based on their size, in ascending order. Second, in each collection, the children's canonization labels are ordered as well. For example, assume we have $3$ collections: $(0, 0)$, $(0, 1)$ and $(2)$. Then those $3$ collections should be ordered as: $(2)$ $<$ $(0,0)$ $<$ $(0,1)$. By using lexicographical order, we can assign the canonization labels in $O(n\log n)$ rather than $O(n^2)$. Algorithm \ref{alg:Canonization} illustrates the details of node canonization.

Figure \ref{fig:Canonization} shows an example of node canonization level by level. At the bottom level ($4$th level) all the nodes have the same canonization label $0$ since they are all leaf nodes without any subtree. When we proceed to the next level ($3$rd level), the nodes in the $3$rd level have different collections of children's canonization labels: $(0)$, $(0,0)$ and $(0,0,0)$. Then we could assign the node with collection of $(0)$ a canonization label $1$. Similarly, the canonization label of node with collection of $(0,0)$ is $2$. We continue like that until we reach the root.

Such node canonization process guarantees that the nodes with the same canonization label in the same level must have isomorphic subtrees. However, the canonization labels of two nodes in different levels do not need to satisfy isomorphism, because we never compare two nodes in different levels. For each level, after the bipartite graph matching process, the canonization labels will be re-assigned. Therefore, for the next level, the node canonization process can be only based on the instant children's canonization re-labels rather than the original labels. 

\subsection{Bipartite Graph Construction}
\label{Bipartite Graph Construction}
To calculate the matching cost, we need to construct a complete weighted bipartite graph and compute the minimum bipartite graph matching.

The weighted bipartite graph $G^2_i$ is a virtual graph. The two node sets of the bipartite graph are the corresponding levels from two $k$-adjacent trees, namely the nodes in  level $L_i(u)$ and level $L_i(v)$. The bipartite graph construction is done after the node padding. Therefore the two node sets $L_i(u)$ and $L_i(v)$ must have the same number of nodes.

$G^2_i$ is a complete weighted bipartite graph which means that for every node pair $(x,y)$ where $x$ $\in$ $L_i(u)$ and $y$ $\in$ $L_i(v)$, there exists a virtual weighted edge. The key component of the bipartite graph construction is to assign the weights to all virtual edges.

\begin{algorithm}
\KwIn{Two levels $L_i(u)$ and $L_i(v)$}
\KwOut{Bipartite graph $G^2_i$}
\ForEach{$(x,y)$, where $x$ $\in$ $L_i(u)$ $\&$ $y$ $\in$ $L_i(v)$}{
       Get collection $S(x)$ $=$ $($ $C(x')$ $|$ $\forall x' \sqsubset x$ $)$\;
       \lIf{$x$ is a padding node}{$S(x)$ $=$ $\emptyset$}
       Get collection $S(y)$ $=$ $($ $C(y')$ $|$ $\forall y' \sqsubset y$ $)$\;
       \lIf{$y$ is a padding node}{$S(y)$ $=$ $\emptyset$}
       $w(x,y)$ $=$ $|S(x) \setminus S(y)|$ $+$ $|S(y) \setminus S(x)|$\;
       \qquad \qquad ($``\setminus"$ is collection difference)\\
       $G^2_i$ $\leftarrow$ $w(x,y)$\;
}

\caption{Bipartite Graph Construction}
\label{alg:Bipartite Graph Construction}
\end{algorithm}

In the complete bipartite graph $G^2_i$, the weight of each edge is decided by the children of two nodes that the edge connects. For two nodes $x$ and $y$, the children of two nodes can be represented as two canonization label collections $S(x)$ and $S(y)$. We denote $S(x)\setminus S(y)$ be the difference between collections $S(x)$ and $S(y)$. The weight $w(x,y)$ is the size of the symmetric difference between the collections $S(x)$ and $S(y)$, i.e. $w(x, y)$ $=$ $|S(x) \setminus S(y)|$ $+$ $|S(y) \setminus S(x)|$. For example, assume node $x$ has the collection $S(x)$ $=$ $(0,0,1)$ and the node $y$ has the collection $S(y)$ $=$ $(0,2)$. The weight between two nodes should be $3$, since node $x$ does not have a child with label $2$, whereas, node $y$ does not have a child with label $1$ and has only one child with label $0$. Notice that, the children canonization label collections allow duplicates, so the difference considers the number of duplicates. If a node is padded, the children canonization label collections is empty as $\emptyset$, since there is no child of a padding node. Algorithm~\ref{alg:Bipartite Graph Construction} shows how to construct the bipartite graph in detail. 

Figure~\ref{fig:Matching} gives an example of constructing the complete weighted bipartite graph. In the figure, let $L_i(u)$ has nodes $A$, $B$ and $C$, while $L_i(v)$ has nodes $X$, $Y$ and $Z$. To avoid confusions, we use Greek characters to represent the canonization labels and integer values to represent the weights. As shown in Figure \ref{fig:Matching}, the weight between node $A$ and $X$ should be $1$ because there is only one child $\beta$ which is the child of $A$ but not the child of $X$.

Since the overall algorithm goes from the bottom level to the root level, when constructing the bipartite graph $G^2_i$ by using $L_i(u)$ and $L_i(v)$, the level $L_{i+1}(u)$ and level $L_{i+1}(v)$ have already been matched which means the canonization labels of nodes in $L_{i+1}(u)$ should be the same as the canonization labels of nodes in $L_{i+1}(v)$.

Moreover, some nodes in  level $L_{i+1}(u)$ or level $L_{i+1}(v)$ may not connect to any node in the level $L_i(u)$ or level $L_i(v)$ respectively. For example, as shown in Figure~\ref{fig:Matching}, the node with canonization label of $\epsilon$ in the left side does not connect to any node of $A$, $B$ or $C$. This is due to the fact that  node $\epsilon$ is a padded node. The padded nodes are not connected to any parent to avoid replicated cost.

The weight between a pair of nodes indicates the number of ``moving a node at the same level" edit operations needed. By calculating the minimum matching, we can get the minimum number of moving edit operations.

\begin{figure}
\centering
\includegraphics[scale=0.7]{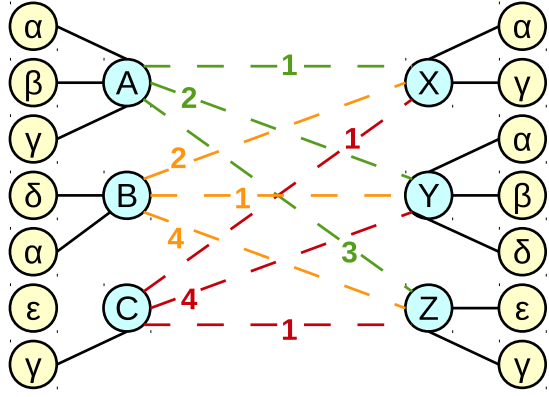}
\caption{Complete Weighted Bipartite Graph}
\label{fig:Matching}
\end{figure}

\subsection{Bipartite Graph Matching}
\label{Bipartite Graph Matching}

After constructing a complete weighted bipartite graph $G^2_i$, we run a minimal bipartite graph matching on $G^2_i$.

The minimal bipartite graph matching is to find a bijective mapping function $f_i$ for $G^2_i$, where $f_i(x)$ $=$ $y$, $x$ $\in$ $L_i(u)$ and $y$ $\in$ $L_i(v)$. The bijective mapping function is to minimize the summation of weights from all nodes in $L_i(u)$ to the corresponding nodes in $L_i(v)$. Let $m(G^2_i)$ be the minimal cost of the matching.
Then,  we have:
\begin{equation}
f_i: \quad m(G^2_i) = Min \sum_{\forall x \in L_i(u)} w(x, f_i(x))
\end{equation}

In this paper, we use the Hungarian algorithm to solve the matching problem. In the bipartite graph matching process, we get the bijective mapping function $f_i$ and the minimal cost $m(G^2_i)$ which we will use in the next step: matching cost calculation. Notice that, the minimal cost for bipartite graph matching is not the matching cost in TED* algorithm.

\subsection{Matching Cost Calculation}
\label{Matching Cost Calculation}

In this section, we show how to calculate the matching cost $M_i$. From the previous section, we get the minimal cost for the bipartite graph matching as $m(G^2_i)$. Let the padding cost from the previous pair of levels $L_{i+1}(u)$ and $L_{i+1}(v)$ be $P_{i+1}$. Then the matching cost $M_i$ can be calculated as:
\begin{equation}
M_i = ( m(G^2_i) - P_{i+1} ) / 2\;
\label{equation:matching cost}
\end{equation}

The matching cost $M_i$ represents the minimal number of ``moving a node at the same level" edit operations needed to convert level $L_i(u)$ to level $L_i(v)$. We consider the following two situations: 1) there is no padding node in level $L_{i+1}(u)$ and level $L_{i+1}(v)$; 2) there exist padding nodes in level $L_{i+1}(u)$ or level $L_{i+1}(v)$;

If there is no padding node in level $L_{i+1}(u)$ and level $L_{i+1}(v)$, every node in $L_{i+1}(u)$ must have a parent in $L_i(u)$. Similarly, every node in $L_{i+1}(v)$ must have a parent in $L_i(v)$. Because the canonization label collections of $L_{i+1}(u)$ and $L_{i+1}(v)$ should be equivalent, then for any node with canonization label $n$ in $L_{i+1}(u)$, there must be a node with canonization label $n$ in $L_{i+1}(v)$. Let $C(n_u) = n$, where $n_u$ $\in$ $L_{i+1}(u)$ and $C(n_v) = n$, where $n_v$ $\in$ $L_{i+1}(v)$. Assume $n_u$ is the child of node $x$ $\in$ $L_{i}(u)$ and $n_v$ is the child of node $y$ $\in$ $L_{i}(v)$. If $f_i(x) = y$, where $f_i$ is the bijective mapping function in the matching, then node $n_u$ and node $n_v$ will not generate any disagreement cost in bipartite graph matching. Otherwise, the pair of nodes $n_u$ and $n_v$ will cost $2$ in the bipartite graph matching, since $x$ is matching to some other node other than $y$ and $y$ is also matching to some other node. However, only one ``moving a node at the same level" edit operation is needed to correct the disagreement, for example, move node $n_v$ from $y$ to $f_i(x)$. Thus, the matching cost should be equal to $m(G^2_i)/2$.

If there are padding nodes in levels $L_{i+1}(u)$ or  $L_{i+1}(v)$, then we can always pad the nodes to the optimal parent which will not generate any matching cost. When we construct the complete bipartite graph $G^2_i$, one padding node in level $L_{i+1}(u)$ or level $L_{i+1}(v)$ will generate $1$ disagreement cost because the padding node is not connected to any node in $L_{i}(u)$ or $L_{i}(v)$, whereas the matched node must connect some node in $L_{i}(v)$ or $L_{i}(u)$. Therefore, the number of ``moving a node at the same level" operations should be $(m(G^2_i)-P_{i+1})/2$.

\subsection{Node Re-Canonization}
\label{Node Re-Canonization} 

The last step for each level is node re-canonization. This step ensures that for each level, only the children information is needed to perform all $6$ steps. Let $f_i$ be the bijective matching mapping function from $L_i(u)$ to $L_i(v)$. Then, $\forall$ $x$ $\in$ $L_i(u)$, $f_i(x) = y$ $\in$ $L_i(v)$. Symmetrically, $\forall$ $y$ $\in$ $L_i(v)$, $f_i^{-1}(y) = x$ $\in$ $L_i(u)$.

In node re-canonization, without loss of generality, we always transform the level with the smaller size to the level with the larger size, since we always do the padding when we calculate the inserting and deleting edit operations. 
Therefore, we have the node re-canonization criteria as follow:
\begin{itemize}
\item If $|L_i(u)|$ $<$ $|L_i(v)|$, $\forall$ $x$ $\in$ $L_i(u)$, $C(x)$ $\Leftarrow$ $C(f_i(x))$
\item If $|L_i(u)|$ $\leq$ $|L_i(v)|$, $\forall$ $y$ $\in$ $L_i(v)$, $C(y)$ $\Leftarrow$ $C(f_i^{-1}(y))$
\end{itemize}

The above node re-canonization criteria means that all the nodes in the level $L_i(u)$ are re-canonized only if the level $L_i(u)$ has padding nodes. The new canonization labels of nodes in $L_i(u)$ are the same as the canonization labels of the matched nodes in $L_i(v)$. Similar re-canonization is needed if the level $L_i(v)$ has padding nodes.

After the node re-canonization, the canonization labels of nodes in level $L_i(u)$ are the same as the canonization labels of nodes in level $L_i(v)$. Then we can proceed to the next pair of levels: $L_{i-1}(u)$ and $L_{i-1}(v)$.

\section{Correctness Proof}
\label{sec: Correctness Proof}

\begin{Lemma}
The Algorithm \ref{alg:TED*_algorithm} correctly returns the TED* distance defined in Definition \ref{Definition:TED*}.
\label{Lemma:Correctness}
\end{Lemma}

\begin{proof}
First, the edit operations defined for TED* cannot change the depth of any existing node. Then, let the series of edit operations with minimal number of edit operations be $E_min = \{ e_1, e_2, ... e_n$. Then each edit operation $e_i$ $\in$ $E$ can be renamed as $e_j^k$, where $k$ is the level of the node in the edit operation and $j$ marks the order of edit operations in level $k$. Therefore, the edit operations in $E_min$ can be grouped by each level.

Second, we prove that the level by level algorithm can return the minimal edit operations each time. In the algorithm, we always pad new nodes to the tree with less nodes.The number of inserting and deleting leaf nodes represents the padding cost in the algorithm. If there exists a ``moving a node" edit operation in one level, it means before moving a node, the bipartite graph matching should count mismatch twice if the mismatched node is not padded. If the mismatched node is padded, when we calculate the weights in the complete bipartite graph, the padded node is not attached to any parent and the bipartite graph matching only counts mismatch once. Therefore, by using Equation \ref{equation:matching cost}, the number of ``moving a node" edit operations can be calculated by using the matching cost.
\end{proof}

\section{Metric Proof}
\label{sec: Metric Proof}

The TED* $\delta_T$ is a metric, if and only if it satisfies $4$ metric properties: non-negativity, symmetry, identity and triangular inequality. Namely, for any  $k$-adjacent trees trees $T(x, k)$, $T(y, k)$ and $T(z, k)$ the following holds:

\begin{tabular}{lr}
$[1]$ $\delta_T(T(x, k),T(y, k))$ $\geq$ $0$ $\qquad$ $\qquad$ & (Non-negativity)\\
$[2]$ $\delta_T(T(x, k),T(y, k))$ $=$ $\delta_T(T(y, k),T(x, k))$ & (Symmetry)\\
$[3]$ $\delta_T(T(x, k),T(y, k))$ $=$ $0$, iff $T(x, k)$ $\simeq$ $T(y, k)$ & (Identity)\\
$[4]$ $\delta_T(T(x, k),T(z, k))$ $\leq$ $\delta_T(T(x, k),T(y, k))$ $+$ $\delta_T(T(y, k),T(z, k))$ & (Triangular Inequality)\\
\end{tabular}

For the identity property, the distance is $0$, if and only if two trees are isomorphic and the roots are mapped by the bijective node mapping function.

In the following sections, we give the detailed proofs for non-negativity, symmetry, identity and triangular inequality. Before the proofs, we rewrite the formula to calculate TED* distance as follows:
\begin{equation}
\delta_T = \sum_{i=1}^{k}(P_i + M_i) = \sum_{i=2}^{k}P_i/2 + \sum_{i=1}^{k-1}m(G^2_i)/2
\label{equation:TED* Calculation}
\end{equation}

The above formula is derived from Algorithm \ref{alg:TED*_algorithm}, where $P_i$ is the padding cost at  level $i$ and  $M_i$ is the matching cost at the level $i$.  We have seen that the matching cost is $M_i$ $=$ $(m(G^2_i)$ $-$ $P_{i+1})$ $/$ $2$. For the root level, there is no padding cost. So $P_1$ $=$ $0$. While for the bottom level, there is no matching cost. So $M_k$ $=$ $0$. Therefore, we get the formula as in Equation (\ref{equation:TED* Calculation}). We will use the equation in the following proofs.

\subsection{Identity}
\label{sec: Identity}

By definition, it is straightforward that the inter-graph node similarity satisfies both non-negativity and symmetry. Since the number of edit operations must be non-negative and both the padding cost and matching cost are non-negative, then TED* must be non-negative. Because all edit operations in TED* can be reverted by  symmetric operations, then TED* is symmetric.

In the following part, we prove that the distance satisfies the identity property as well, where  $\delta_T$ is the TED*.

\begin{Lemma}
$\delta_T(T(x, k),T(y, k))$ $=$ $0$, iff $T(x, k)$ $\simeq$ $T(y, k)$
\label{Lemma:Identity}
\end{Lemma}

\begin{proof}
If the TED* $\delta_T(T(x, k),T(y, k))$ between two trees is $0$, there is no edit operation needed to convert $T(x, k)$ to an isomorphic tree of $T(y, k)$. Then the two trees $T(x, k)$ and $T(y, k)$ are isomorphic.

If two $k$-adjacent trees $T(x, k)$ and $T(y, k)$ are isomorphic, there exists a bijective mapping function $f$ from all nodes in tree $T(x, k)$ to the nodes in $T(y, k)$. Then, in each level, number of nodes from two trees should be the same. Then the padding cost is $0$ for each level. Also in each level, the bijective mapping function $f$ makes the bipartite graph matching to return $0$. Therefore the matching cost is $0$. Thus, for a pair of isomorphic trees, the TED* must be $0$.
\end{proof}

\subsection{Triangular Inequality}
\label{sec: Triangular Inequality}

In this section, we prove that the TED* satisfies the triangular inequality. Let $x$, $y$ and $z$ be three trees for short. According to Equation (\ref{equation:TED* Calculation}), we have
\begin{equation}
  \left\{
 	\begin{aligned}
		\delta_T(x,y) = \sum_{i=2}^{k}P_i^{xy}/2 + \sum_{i=1}^{k-1}m(G^2_i)^{xy}/2\\
		\delta_T(y,z) = \sum_{i=2}^{k}P_i^{yz}/2 + \sum_{i=1}^{k-1}m(G^2_i)^{yz}/2\\
		\delta_T(x,z) = \sum_{i=2}^{k}P_i^{xz}/2 + \sum_{i=1}^{k-1}m(G^2_i)^{xz}/2\\
	\end{aligned}
  \right.
\end{equation}

In order to prove $\delta_T(x,z)$ $\leq$ $\delta_T(x,y)$ $+$ $\delta_T(y,z)$,  we can prove that the following two inequalities hold for each level $i$:
\begin{align}
P_i^{xz} & \leq P_i^{xy} + P_i^{yz}\label{Equation:Padding Inequality}\\
m(G^2_i)^{xz} & \leq m(G^2_i)^{xy} + m(G^2_i)^{yz} \label{Equation:Matching Inequality}
\end{align}

Inequality (\ref{Equation:Padding Inequality}) means that the padding cost satisfies the triangular inequality and Inequality (\ref{Equation:Matching Inequality}) represents the minimal cost for the bipartite graph matching and satisfies the triangular inequality.

First of all, we prove that the padding cost satisfies the triangular inequality.
\begin{proof}
Let $L_i(x)$ be the $i$th level of $k$-adjacent tree extracted from node $x$. Similarly, $L_i(y)$ and $L_i(z)$ are the levels for nodes $y$ and $z$ respectively. According to Algorithm \ref{alg:TED*_algorithm}, $P_i$ $=$ $\big||L_i(x)|$ - $|L_i(y)|\big|$. Then we have:

\begin{align*}
P_i^{xz}  &= \big||L_i(x)| - |L_i(z)|\big|\\
			&= \big|(|L_i(x)| - |L_i(y)|) - (|L_i(z)| - |L_i(y)|)\big|\\
			&\leq \big||L_i(x)| - |L_i(y)|\big| + \big||L_i(z)| - |L_i(y)|\big|\\
			&= P_i^{xy} + P_i^{yz}
\end{align*}

Therefore, Inequality (\ref{Equation:Padding Inequality}) holds.
\end{proof}

Next, we prove that the minimal cost for bipartite graph matching satisfies the triangular inequality.
\begin{proof}
Let $L_i(x)$ be the $i$th level of $k$-adjacent tree extracted from node $x$. Similarly, $L_i(y)$ and $L_i(z)$ are the levels for nodes $y$ and $z$ respectively.

Let $f$ be the bijective mapping function from level $L_i(x)$ to level $L_i(z)$ which satisfies the minimal bipartite graph matching. Similarly, let $g$ and $h$ be the bijective mapping functions from level $L_i(x)$ to level $L_i(y)$ and from level $L_i(y)$ to level $L_i(z)$. Then, for any node $\alpha$ $\in$ $L_i(x)$, we have $f(\alpha)$ $\in$ $L_i(z)$. Also, for any node $\alpha$ $\in$ $L_i(x)$, we have $g(\alpha)$ $\in$ $L_i(y)$ and for any node $\beta$ $\in$ $L_i(y)$, we have $h(\beta)$ $\in$ $L_i(z)$.

According to Algorithm \ref{alg:TED*_algorithm}, we can rewrite the minimal cost for bipartite graph matching $m(G^2_i)^{xz}$, $m(G^2_i)^{xy}$ and $m(G^2_i)^{yz}$ as follows:
\begin{equation}
  \left\{
 	\begin{aligned}
	m(G^2_i)^{xz} = \sum w(\alpha,f(\alpha))\\
	m(G^2_i)^{xy} = \sum w(\alpha,g(\alpha))\\
	m(G^2_i)^{yz} = \sum w(\beta,h(\beta))\\
	\end{aligned}
  \right.
\end{equation}

First we prove that the weights in three bipartite graphs satisfy the triangular inequality:
\begin{equation}
w(\alpha,\gamma) \leq w(\alpha,\beta) + w(\beta,\gamma)
\label{Equation:Weights Inequality}
\end{equation}

Since the weights are calculated  using $w(x,y)$ $=$ $|S(x) \setminus S(y)|$ $+$ $|S(y) \setminus S(x)|$ according to Algorithm \ref{alg:TED*_algorithm}, the Inequality (\ref{Equation:Weights Inequality}) can be transformed to:
\begin{equation}
\begin{aligned}
|S(\alpha)\setminus S(\gamma)| +|S(\gamma)\setminus S(\alpha)| \leq |S(\alpha)\setminus S(\beta)|+|S(\beta)\setminus S(\alpha)| +|S(\beta)\setminus S(\gamma)|+|S(\gamma)\setminus S(\beta)|
\end{aligned}
\label{Equation:Weights Inequality Transformation}
\end{equation}

Let $e$ be an element. If $e$ $\in$ $S(\alpha)$ and $e$ $\notin$ $S(\gamma)$, then $e$ costs one disagreement for $|S(\alpha) \setminus S(\gamma)|$. If $e$ $\notin$ $S(\alpha)$ and $e$ $\in$ $S(\gamma)$, then $e$ costs one disagreement for $|S(\gamma) \setminus S(\alpha)|$. For both situations, consider whether $e$ $\in$ $S(\beta)$ or $e$ $\notin$ $S(\beta)$. We can see that for all situations, Inequality (\ref{Equation:Weights Inequality Transformation}) holds.

Then we prove the inequality (\ref{Equation:Matching Inequality}). Since $f$, $g$ and $h$ are all bijective mapping functions, so we know for any node $\alpha$ $\in$ $L_i(x)$, both $f(\alpha)$ $\in$ $L_i(z)$ and $h(g(\alpha))$ $\in$ $L_i(z)$ hold. Then according to Inequality (\ref{Equation:Weights Inequality}), we have:
\begin{equation}
w(\alpha,h(g(\alpha))) \leq w(\alpha,g(\alpha)) + w(g(\alpha),h(g(\alpha)))
\end{equation}

Because $m(G^2_i)^{xz}$ $=$ $\sum w(\alpha,f(\alpha))$ is the minimal matching so we have:

\begin{align*}
m(G^2_i)^{xz}  &= \sum w(\alpha,f(\alpha))\\
				 &\leq \sum w(\alpha,h(g(\alpha)))\\
				 &\leq \sum w(\alpha,g(\alpha)) + \sum w(g(\alpha),h(g(\alpha)))\\
				 &= m(G^2_i)^{xy} + m(G^2_i)^{yz}
\end{align*}

Therefore, the Inequality (\ref{Equation:Matching Inequality}) is proved.
\end{proof}

Because both Inequality (\ref{Equation:Padding Inequality}) and Inequality (\ref{Equation:Matching Inequality}) hold for each level, we can prove that the TED* satisfies the triangular inequality overall.

\section{Isomorphism Computability}
\label{sec: Isomorphism Computability}

In this section, we discuss why we chose to use the neighborhood tree to represent a node's neighborhood topological structure rather than the neighborhood subgraph.

We define the node identity property as follow:  
\begin{Definition}
For two nodes $u$ and $v$, $u \equiv v$ if and only if  $\delta(u,v)$ $=$ $0$.
\label{Node Equivalency}
\end{Definition}

The above definition means that two nodes are equivalent if and only if, the distance between the two nodes is $0$. In our NED distance between inter-graph nodes, the neighborhood $k$-adjacent trees are the signatures of nodes. Therefore, the rooted tree isomorphism should represent the node equivalency. However, if the $k$-hop neighborhood subgraphs are the signatures of nodes,  subgraph isomorphism should represent the node equivalency. It is easy to prove that to satisfy the identity in the metric properties, the computation problem of node similarity with $k$-hop neighborhood subgraphs is as hard as graph isomorphism, which belongs to class $\NP$, but not known to belong to class $\P$.

\begin{Lemma}
Given two nodes $u$ $\in$ $G_u$, $v$ $\in$ $G_v$ and a value $k$. Let $\delta(u, v)$ be a distance between nodes $u$ and $v$ , where $\delta(u, v)$ $=$ $0$ if and only if two $k$-hop neighborhood subgraphs $G_s(u, k)$ and $G_s(v, k)$ are isomorphic and $v$ $=$ $f(u)$, where $f$ is the bijective node mapping that makes two graphs isomorphic. Then, the computation of distance function $\delta$ is at least hard as graph isomorphism problem.
\end{Lemma}

\begin{proof}
There exists a polynomial-time reduction from the graph isomorphism problem to the computation problem of distance function $\delta(u, v)$. Given two graphs $G_u$ and $G_v$, we can add two new nodes $u$ and $v$ to $G$ and $G_v$ respectively. Let node $u$ connect all nodes in $G$ and $v$ connect all nodes in $G_v$. Then two graphs $G_u$ and $G_v$ are converted to two $1$-hop neighborhood subgraphs rooted at $u$ and $v$ denoted as $G_s(u,1)$ and $G_s(v, 1)$ separately. If and only if the distance $\delta(u, v)$ $=$ $0$, $G_u$ and $G_v$ can be isomorphic. So, the computation of the distance function can verify the graph isomorphism. Therefore the computation problem of distance function should be at least as hard as graph isomorphism. If the distance function can be computed in polynomial-time, it means that the graph isomorphism can be verified in polynomial time by using the distance function: if the distance is $0$, two graphs are isomorphic, otherwise not. However, the graph isomorphism does not known to belong to class $\P$.
\end{proof}

Lemma 2 guarantees that no matter what kind of distance functions is chosen, if the nodes are compared based on their $k$-hop neighborhood subgraphs, the distance cannot be both polynomially computable and satisfy the identity property. Therefore, in order to have a method that is polynomially computable and satisfy the metric properties we chose to use trees and not graphs.

Actually, graph edit distance (GED) is a metric which can be a distance function to compare inter-graph nodes based on $k$-hop neighborhood subgraphs. Unfortunately, due to the computability issue proved above, the computation of graph edit distance is known to belong to $\NP$-Hard problems\cite{DBLP:journals/pvldb/ZengTWFZ09}.

Notice that Jin et al. \cite{DBLP:conf/kdd/JinLH11} proposed a set of axiomatic role similarity properties for intra-graph node measures. The major difference between axiomatic properties and metric properties in this paper is that: in the set of axiomatic role similarity properties, the identity is verified in single direction. Namely, if two nodes are automorphic, the distance is $0$. Whereas, if the distance is $0$, two nodes may not be automorphic. The reason of single direction verification is because the node automorphism is chosen to represent the node equivalency. Whereas, the graph automorphism problem is also unknown to belong to class $\P$.

Therefore, in this paper, we choose tree structure to represent the neighborhood topology of a node. Because the tree isomorphism problem can be solved in polynomial time even for unlabeled unordered trees. It becomes possible to construct a polynomial-time computable distance function for comparing inter-graph nodes based on neighborhood trees and at the same time the node distance satisfies the node identity in Definition \ref{Node Equivalency}. Notice that, any spanning tree extracted from the neighborhood can be the representation of a node. In this paper we adopt the $k$-adjacent tree as an example, since $k$-adjacent tree can be deterministically extracted and according to Section \ref{sec:TED*, TED and GED Comparison}, it shows that $k$-adjacent tree is able to precisely capture the neighborhood topological information of a node.

\section{Complexity Analysis}
\label{sec: Complexity Analysis}
The TED* algorithm in Algorithm \ref{alg:TED*_algorithm} is executed level by level and includes $6$ steps sequentially.  Let the number of levels be $k$ and the size of level $L_i(u)$ and level $L_i(v)$ be $n$.

The node padding can be executed in $O(n)$ time and the node canonization can be calculated in $O(n \log n)$ time in our Algorithm \ref{alg:Canonization}. The bipartite graph construction needs $O(n^2)$ time to generate all weights for a completed bipartite graph. The most time consuming part is the bipartite graph matching. We use the improved Hungarian algorithm to solve the bipartite graph matching problem with time complexity  $O(n^3)$. The matching cost calculation can be executed in constant time and node re-canonization is in $O(n)$.

Clearly, the time complexity is dominant by the bipartite graph matching part which cost $O(n^3)$. Notice that, $n$ is the number of nodes per level. Therefore, the overall time complexity of computing TED* should be $O(kn^3)$. Indeed, for the real-world applications, we show that the parameter $k$ will not be large. Therefore, the TED* can be computed efficiently in practice.

\section{Parameter K and Monotonicity}
\label{sec: Parameter K and Monotonicity}

In NED,  there is only one parameter $k$ which represents how many levels of neighbors should be considered in the comparison. There exists a monotonicity property on the distance and the parameter k in NED. Let $u$ and $v$ be two nodes and let $T(u, k)$ and $T(v, k)$ be two $k$-adjacent trees of those two nodes respectively. Denote $\delta_T$ as the TED* between two $k$-adjacent trees. The monotonicity property is defined as follows:

\begin{Lemma}
$\delta_T(T(u, x), T(v, x))$ $\leq$ $\delta_T(T(u, y), T(v, y))$, $\forall x, y >0 $ and $x \leq y$
\label{Lemma:Monotonicity}
\end{Lemma}

\begin{proof}
The proof of Lemma \ref{Lemma:Monotonicity} is based on the procedures in Algorithm \ref{alg:TED*_algorithm}. In Lemma \ref{Lemma:Monotonicity} , $x$ and $y$ are total levels for $k$-adjacent trees. According to Equation \ref{equation:TED* Calculation}, $\delta_T(T(u, y), T(v, y))$ can be rewritten  as:
\begin{equation}
\delta_T(T(u, y), T(v, y)) = \sum_{i=1}^{y}(P^y_i + M^y_i) ;
\end{equation}

While $\delta_T(T(u, x), T(v, x))$ can be rewritten  as:
\begin{equation}
\delta_T(T(u, x), T(v, x)) = \sum_{i=1}^{x}(P^x_i + M^x_i) ;
\end{equation}

Since the padding cost and matching cost are non-negative in each level as proven in Section \ref{sec: Identity}.  It is obvious that
\begin{equation}
\sum_{i=x}^{y}(P^y_i + M^y_i) \geq 0
\end{equation}

Then we try to prove that
\begin{equation}
\sum_{i=1}^{x}(P^y_i + M^y_i) \geq \sum_{i=1}^{x}(P^x_i + M^x_i)
\end{equation}

According to the algorithm, for the levels from $1$ to $x$, $P^y_i $ $=$ $P^x_i $, since the top $x$ levels of $T(u, y)$ and $T(v, y)$ should have the same topological structures as $T(u, x)$ and $T(v, x)$ respectively. Meanwhile, we have $M^y_x$ $\geq$ $M^x_x$, because at the $x$th level, the children of nodes in $T(u, y)$ and $T(v, y)$ may have different canonization labels, but the nodes in $T(u, x)$ and $T(v, x)$ are all leaf nodes. So the matching cost between $T(u, x)$ and $T(v, x)$ at the $x$th level should not be larger than the matching cost between $T(u, y)$ and $T(v, y)$ at the $x$th level. For all the levels above the $x$th level, the matching cost for two distances should be the same.

Therefore, for a given pair of nodes $u$ and $v$, if $x$ $\leq$ $y$, then $\delta_T(T(u, x), T(v, x))$ cannot be larger than $\delta_T(T(u, y), T(v, y))$
\end{proof}

The monotonicity property is useful for picking the parameter $k$ for specific tasks as follows: the node similarity, NED, for a smaller parameter $k$ is a lower bound of  NED for a larger parameter $k$. Then, for nearest neighbor similarity node queries,  increasing $k$ may reduce the number of ``equal" nearest neighbor nodes in the result set. For top-$l$ similarity node ranking, increasing $k$ may break the ties in the rank. In Section \ref{sec:Parameter K Analysis}, we show how the monotonicity property affects the query quality using real world datasets.

\section{TED*, TED and GED}
\label{sec:TED*, TED and GED}
In this section we briefly discuss the differences among modified tree edit distance, TED*, tree edit distance (TED) and graph edit distance (GED) from two aspects: edit operations and edit distances.

In unlabeled TED, there are only two types of edit operations: insert a node and delete a node. The edit operations in TED guarantee that no matter which operation is applied to a tree, the result is still a tree structure. However, the edit operations in unlabeled GED are different. The edit operations in GED are categorized into node operations and edge operations. Therefore in GED there can be an edge insertion or deletion without changing any node. While only isolated nodes can be inserted or deleted in GED. Different from TED and GED, TED* has another set of edit operations: ``Inserting a leaf node'', ``Deleting a leaf node'' and ``Moving a node in the same level''. All edit operations in TED* preserve the tree structure which is the same as TED. Moreover, similar to TED, TED* also has edit operations on nodes only.

Indeed, the edit operations in TED and TED* can be represented as a series of edit operations in GED. For example, inserting a node in TED is equivalent to inserting an isolated node and several edges which connect to the new node. There is no one to one mapping between an edit operation in TED and a series of edit operations in GED. However, there exists a one-to-one mapping from the edit operation in TED* to exactly two edit operations in GED. ``Inserting a leaf node'' in TED* is equivalent to inserting an isolated node and inserting an edge in GED. The inserted edge connects the new node and an existing node. ``Deleting a leaf node'' in TED* is equivalent to deleting an edge and deleting an isolated node in GED. ``Moving a node in the same level'' in TED* can be represented as deleting an edge and inserting an edge in GED. The node is moved by changing its parent node which is the same as detaching from the former parent node and attaching to the new parent node. Therefore, by applying GED on tree structures, we have the following bound:
\begin{equation}
\delta_{GED}(t_1, t_2) \leq 2 \times \delta_T(t_1, t_2)
\label{equation:GED_TED*}
\end{equation}
where $\delta_{GED}$ is the distance on GED and $\delta_T$ is the distance on TED*.

When calculating the edit distance between a pair of unlabeled unordered trees, the edit distance in TED* can be smaller or larger than the edit distance in TED. Since TED allows to insert or delete intermediate node (with parent and children) in a tree and such one operation should be done by a series of edit operations in TED*. TED may be smaller than TED*. Whereas since TED* allows to move a node in the same level but TED does not have a such edit operation, TED has to use a series of edit operations to match one  ``Moving a node in the same level'' operation in TED*. Therefore, TED may be larger than TED*. Indeed, according to our experiment in Section \ref{sec:TED*, TED and GED Comparison}, we show that TED* is pretty closed to TED even though they have different set of edit operations. In Section \ref{sec:Weighted TED*}, we propose a weighted TED* which can be an upper-bound of TED.

\section{Weighted TED*}
\label{sec:Weighted TED*}

In the TED* of this paper, all the edit operations have the same cost as $1$. However, sometimes, maybe different edit operations should be assigned to different costs. Indeed, the nodes at different levels of the $k$-adjacent tree should have different impacts on TED* distance and the nodes which are more closed to the root should play more important roles. In this section, we introduce the weighted TED* and prove that the weighted TED* is still a metric.

Let the ``Inserting a leaf node'' operation and the``Deleting a leaf node'' operations at the level $i$ have the weights of $w^1_i$. The ``Moving a node in the same level'' operations at the level $i$ have the weights of $w^2_i$. Therefore, we can calculate the weighted TED* according to Equation \ref{equation:TED* Calculation}. Then the weighted TED* can be rewritten as $\delta_{T(W)}$:
\begin{equation*}
\begin{aligned}
\delta_{T(W)} &=  \sum_{i=1}^{k}(w^1_i * P_i + w^2_i * M_i) \\
		     &= \sum_{i=2}^{k} w^1_i * P_i/2 + \sum_{i=1}^{k-1}w^2_i * m(G^2_i)/2
\end{aligned}	     
\end{equation*}

\begin{Lemma}
$\delta_{T(W)}$  is a metric if $\forall$ $w^1_i $ $>$ $0$, $\forall$ $w^2_i $ $>$ $0$
\label{Lemma:Weighted TED*}
\end{Lemma}

\begin{proof}
The proofs of non-negativity, symmetry and node equivalency should be the same as the proofs for the original TED* in Section \ref{sec: Metric Proof}. Since all weights are positive, the non-negativity and equivalency still preserve.

For the triangular inequality, the Inequality \ref{Equation:Padding Inequality} and Inequality \ref{Equation:Matching Inequality} still hold as:
\begin{align}
&P_i^{xz} \leq P_i^{xy} + P_i^{yz}\\
&m(G^2_i)^{xz} \leq m(G^2_i)^{xy} + m(G^2_i)^{yz}
\end{align}

The proofs of above two equations are the same as in Section \ref{sec: Metric Proof}. Therefore, we still have:
\begin{equation}
\delta_{T(W)} (x,z) \leq \delta_{T(W)} (x,y) + \delta_{T(W)} (y,z)
\end{equation}
\end{proof}

As discussed in Section \ref{sec:TED*, TED and GED}, because the edit operations in TED* and TED are different, TED* may be smaller or larger than TED. By assigning the weights for different edit operations in TED*, we can provide an upper-bound for TED.
\begin{Definition}
$\delta_{T(W+)} $ $=$ $\sum_{i=1}^{k}(P_i + 4 * i * M_i) $
\label{Definition:TED-W+}
\end{Definition}

Since $\delta_{T(W+)}$ is a weighted TED* where $w^1_i $ $=$ $1$ and $w^2_i $ $=$ $4*i$, then $\delta_{T(W+)}$ is still a metric. Furthermore, $\delta_{T(W+)}$ is also an upper-bound for TED.

\begin{Lemma}
$\delta_{T(W+)}(T_1, T_2)$  $\geq$ $\delta_{TED}(T_1, T_2)$
\label{Lemma:TED-W+}
\end{Lemma}

\begin{proof}
To prove the Lemma \ref{Lemma:TED-W+},  we need to prove that the series of edit operation for $\delta_{T(W+)}(T_1, T_2)$ can be translated into a valid series of edit operations in TED and the cost for $\delta_{T(W+)}(T_1, T_2)$ is equal to number of edit operations in TED.

Firstly, since all operations of ``inserting a leaf node" and ``deleting a leaf node" in $\delta_{T(W+)}(T_1, T_2)$ have the same weight of $1$ which are the valid edit operations in TED. Therefore, all inserting and deleting operations in $\delta_{T(W+)}(T_1, T_2)$ can be one-to-one translated to edit operations in TED.

The ``moving a node" operation in $\delta_{T(W+)}(T_1, T_2)$ can be translated into a series of node edit operations in TED. Assume the node to be moved is at the $i$th level, we can use $4*i$ edit operations in TED to move this node by deleting former parent nodes to the root and inserting new parent nodes from the root back to the $i$th level. At the same time, the previous parents are preserved. Therefore, every moving operation in $\delta_{T(W+)}(T_1, T_2)$ can be translated to $4*i$ edit operations in TED.

Then there exists a valid series of edit operations in TED which transforms $T_1$ to $T_2$ and the cost is $\delta_{T(W+)}(T_1, T_2)$. Then the TED distance between $T_1$ and $T_2$ is no larger than $\delta_{T(W+)}(T_1, T_2)$.
\end{proof}

\section{Experiments}
\label{Experiments}

In this section, we empirically evaluate the efficiency and effectiveness of the inter-graph node similarity based on edit distance. All experiments are conducted on a computing node with 2.9GHz CPU and 32GB RAM running 64-bit CentOS 7 operating system. All comparisons are implemented in Java.

In particular, we evaluate TED* and NED over $5$ aspects: 1) Compare the efficiency and distance values of TED* against the original tree edit distance and the graph edit distance; 2) Evaluate the performance  of TED* with different sizes of trees; 3) Analyze the effect of parameter $k$ on NED; 4) Compare the computation and nearest neighbor query performance of NED with HITS-based similarity and Feature-based similarity; and 5) Provide a case study for graph de-anonymization.

We denote the original tree edit distance on unordered trees as TED, the graph edit distance as GED, the HITS-based similarity measurement as HITS, and the Feature-based similarity as Feature. Notice that, the Feature-based similarity here means ReFeX. The OddBall and NetSimile are simplified versions of ReFeX with parameter $k = 1$.

The datasets used in the experiments are real-world graphs that come from the KONECT~\cite{DBLP:conf/www/Kunegis13} and SNAP~\cite{snapnets} datasets. Table \ref{Table:Datasets Summary} summarizes the statistical information of the $6$ graphs and the abbreviations that we use.

All the distances in the experiments are computed using  pair of nodes from two different graphs to verify the ability of inter-graph node similarity.

\begin{table}
\centering
\caption{Datasets Summary}
\begin{tabular}{|l|r|r|} \hline
Dataset & $\mathsmaller{\#}$ Nodes & $\mathsmaller{\#}$ Edges\\ \hline \hline
CA Road (CAR) & 1,965,206 & 2,766,607\\ \hline
PA Road (PAR) & 1,088,092 & 1,541,898\\ \hline
Amazon (AMZN) & 334,863 & 925,872\\ \hline
DBLP (DBLP) & 317,080 & 1,049,866\\ \hline
Gnutella (GNU) & 62,586 & 147,892\\ \hline
Pretty Good Privacy (PGP) & 10,680 & 24,316\\ \hline
\end{tabular}
\label{Table:Datasets Summary}
\end{table}

\subsection{TED*, TED and GED Comparison}
\label{sec:TED*, TED and GED Comparison}

\begin{figure*}
\centering
	\begin{subfigure}[b]{0.48\textwidth}
		\epsfig{file=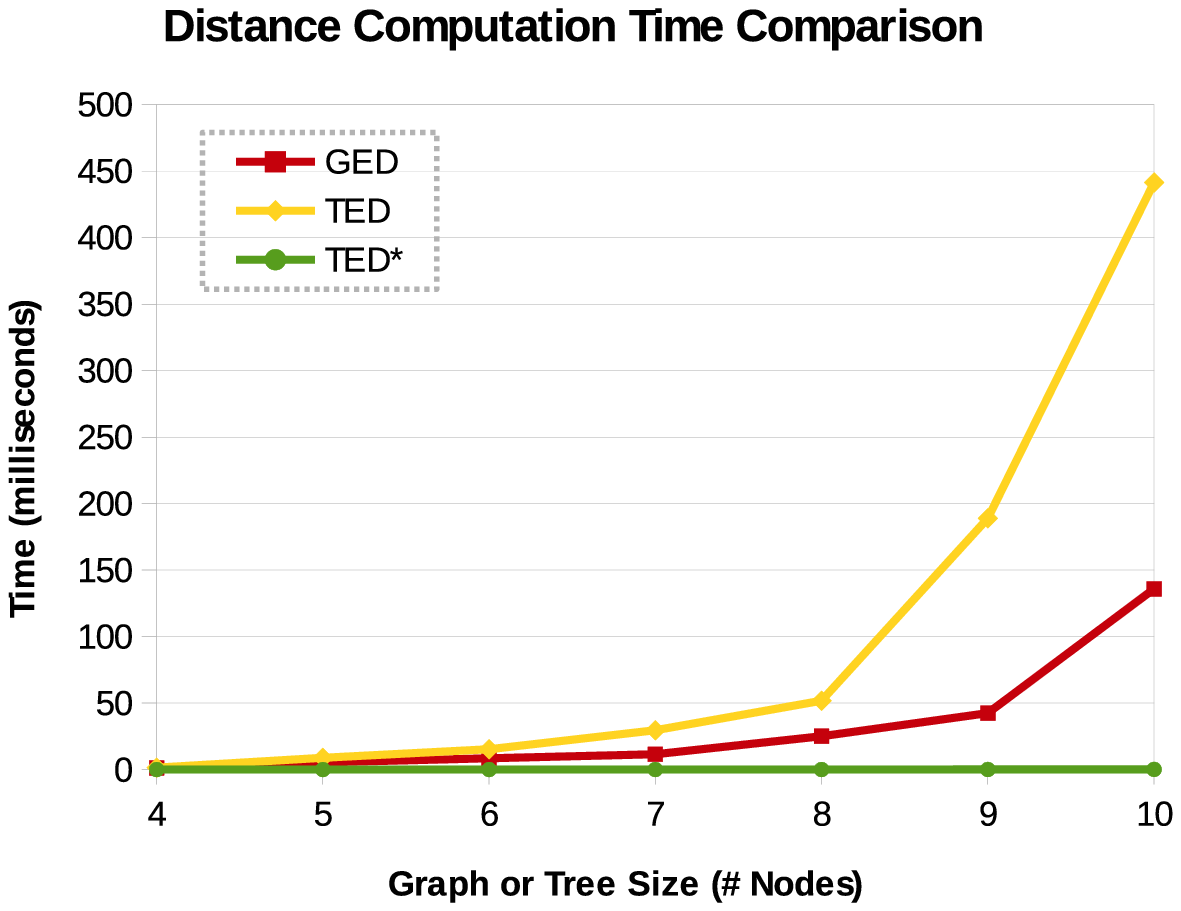,scale=0.6}
		\caption{Computation Time}
		\label{fig:Comparisons Time}
	\end{subfigure}
	\begin{subfigure}[b]{0.48\textwidth}
		\epsfig{file=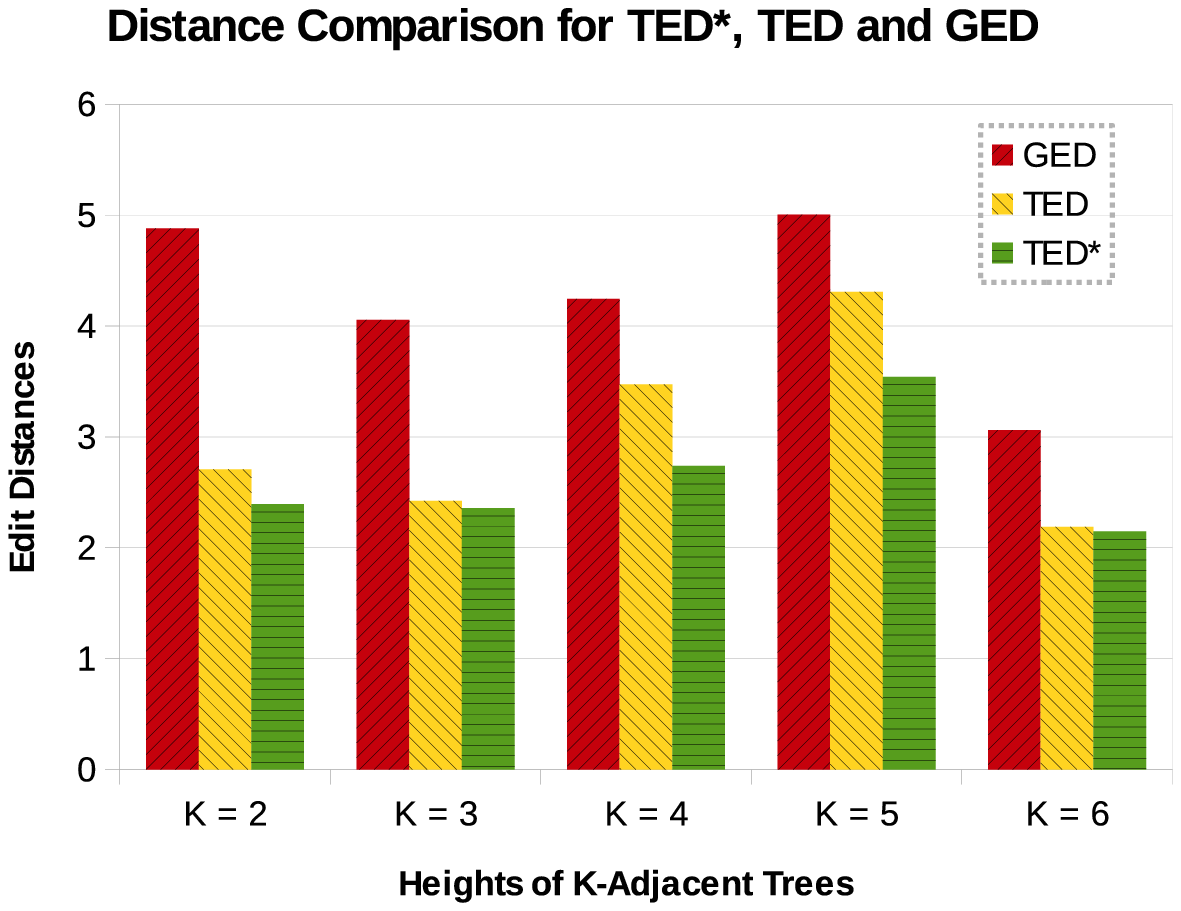,scale=0.6}
		\caption{Distance Values}
		\label{fig:Distance Values}
	\end{subfigure}
\caption{Comparisons among TED*, TED and GED}
\label{fig:Comparisons among TED*, TED and GED}
\end{figure*}

\begin{figure*}
\centering
	\begin{subfigure}[b]{0.48\textwidth}
		\epsfig{file=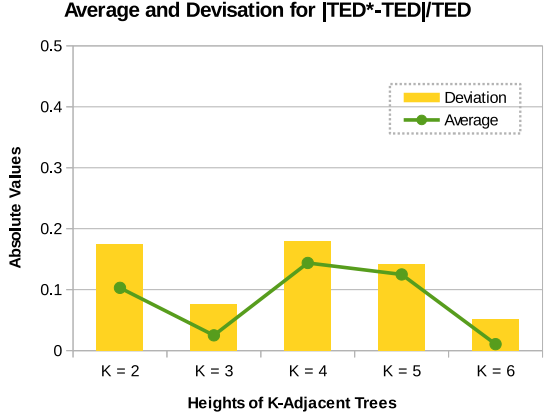,scale=0.6}
		\caption{Distance Difference}
		\label{fig:Distance Difference}
	\end{subfigure}
	\begin{subfigure}[b]{0.48\textwidth}
		\epsfig{file=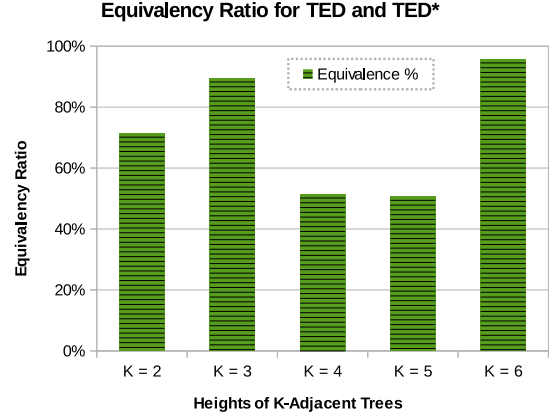,scale=0.6}
		\caption{Equivalency Ratio}
		\label{fig:Equivalency Ratio}
	\end{subfigure}
\caption{Comparisons among TED*, TED and GED Cont.}
\label{fig:Comparisons among TED*, TED and GED Cont.}
\end{figure*}

In this section, we check the efficiency of TED* and compare how close is the computed distance to TED on the same trees and the GED distance on the k-hop subgraphs around the same nodes.

In Figure \ref{fig:Comparisons Time}, we show the computation time of TED*, TED, and GED. TED* is not only polynomially computable but is also very efficient. On the other hand, computing the exact values for GED and TED on unordered trees is expensive since they are NP-Complete problems. The most widely used method for computing them is using the A*-based  algorithm \cite{DBLP:journals/tssc/HartNR68}. As illustrated in \cite{DBLP:conf/sspr/NeuhausRB06, DBLP:journals/pvldb/ZengTWFZ09},  this method can only deal with small graphs and trees with only up to 10-12 nodes. As the number of nodes increases, the searching tree in A* algorithm grows exponentially. However, TED* is able to compare unordered trees up to hundred nodes in milliseconds as shown in Section \ref{sec:TED* and NED Computation}. In this experiment, $400$ pairs of nodes are randomly picked from (CAR) and (PAR) graphs respectively. The $k$-adjacent trees and $k$-hop subgraphs are extracted to conduct TED*, TED, and GED.

In Figure \ref{fig:Distance Values}, we show the distance values for TED*, TED, and GED. It is clear that, by using $k$-adjacent trees, TED* and TED are able to compare the neighborhood structures between two nodes. As shown, TED* is slightly smaller than TED in some cases because TED* has a ``move a node" operation which should be done by a series of edit operations in TED. Moreover, when parameter $k$ is relatively small for $k$-adjacent trees, there are quite few intermediate-node insertions and deletions for TED.

In Figure \ref{fig:Distance Difference}, we show the difference between TED* and TED in details. The figure shows the average and standard deviation of the relative error between TED* and TED. The difference is calculated by 
\begin{align*}
|TED - TED*|/TED
\end{align*}

The average is between $0.04$ to  $0.14$ and the deviation below $0.2$. This means that in most cases the TED and TED* values are almost the same.

In Figure \ref{fig:Equivalency Ratio} we show for how many pairs the TED* are {\bf exactly} the same as TED. It is clear that for most $k$ more than $50\%$ of the cases the two distances are exactly the same and in some cases this can get to more than $80\%$.  This shows that TED* and TED are actually very close on the trees that we were able to compute the exact value for TED.

\subsection{TED* and NED Computation}
\label{sec:TED* and NED Computation}

\begin{figure*}
\centering
	\begin{subfigure}[b]{0.48\textwidth}
		\epsfig{file=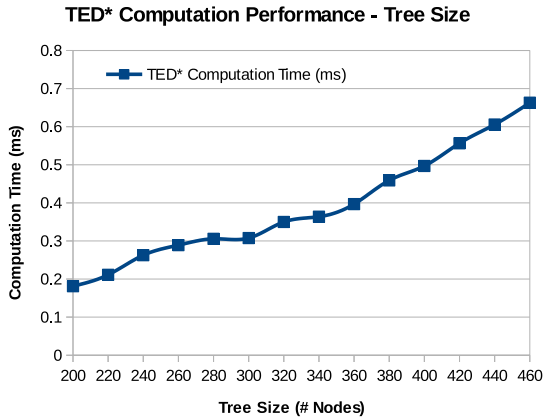,scale=0.6}
		\caption{TED* Computation Time}
		\label{fig:TED* Computation Time}
	\end{subfigure}
	\begin{subfigure}[b]{0.48\textwidth}
		\epsfig{file=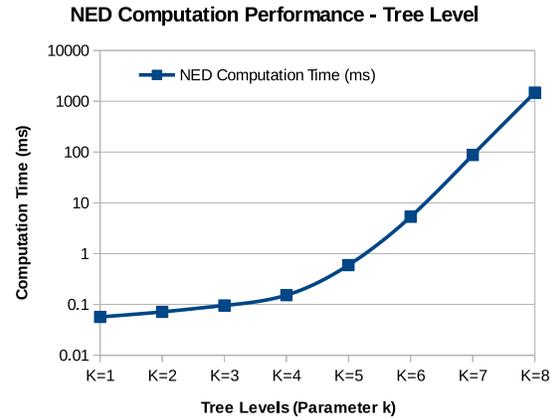,scale=0.6}
		\caption{NED Computation Time}
		\label{fig:NED Computation Time}
	\end{subfigure}
\caption{Computation Time of TED* and NED}
\label{Fig:Computation Time of TED* and NED}
\end{figure*}

\begin{figure*}
\centering
	\begin{subfigure}[b]{0.48\textwidth}
		\epsfig{file=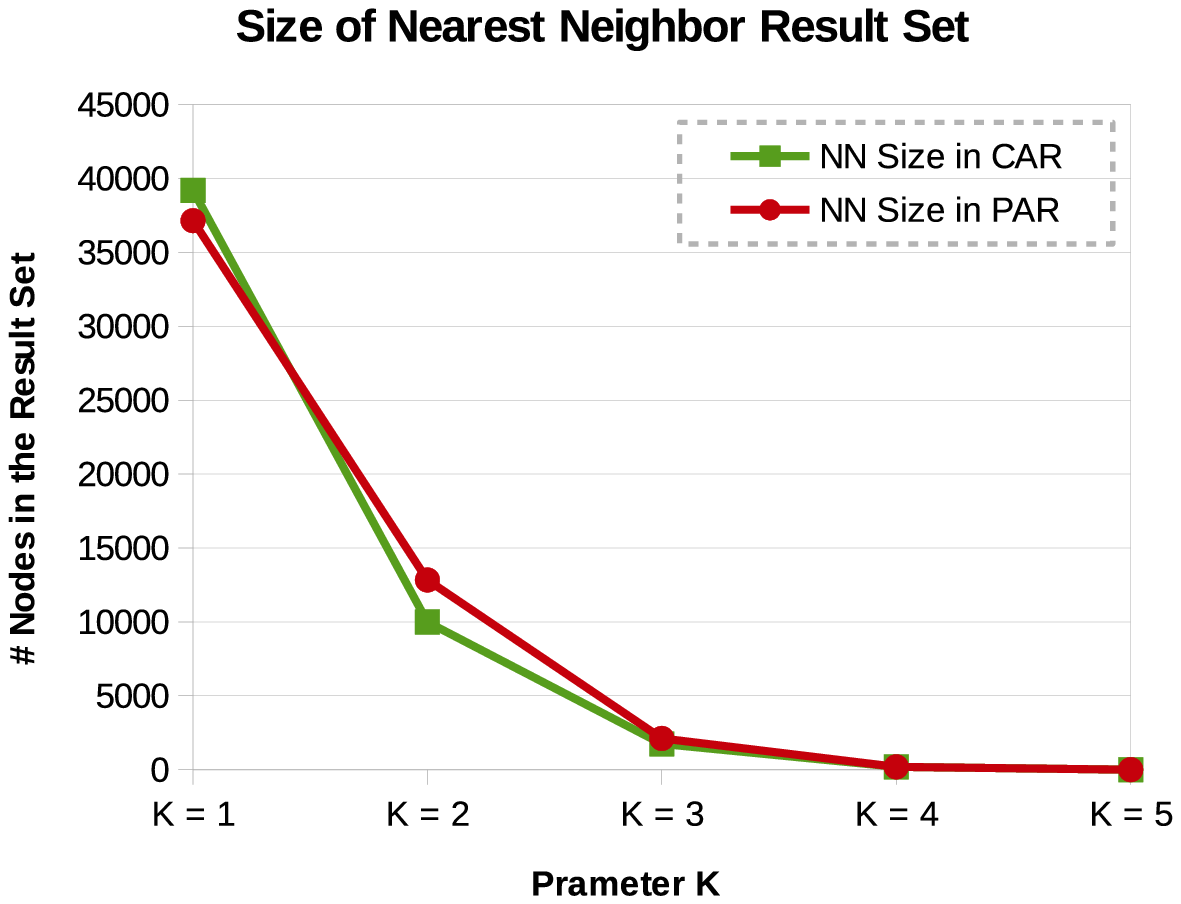,scale=0.6}
		\caption{Nearest Neighbor Query}
		\label{fig:K_NN}
	\end{subfigure}
	\begin{subfigure}[b]{0.48\textwidth}
		\epsfig{file=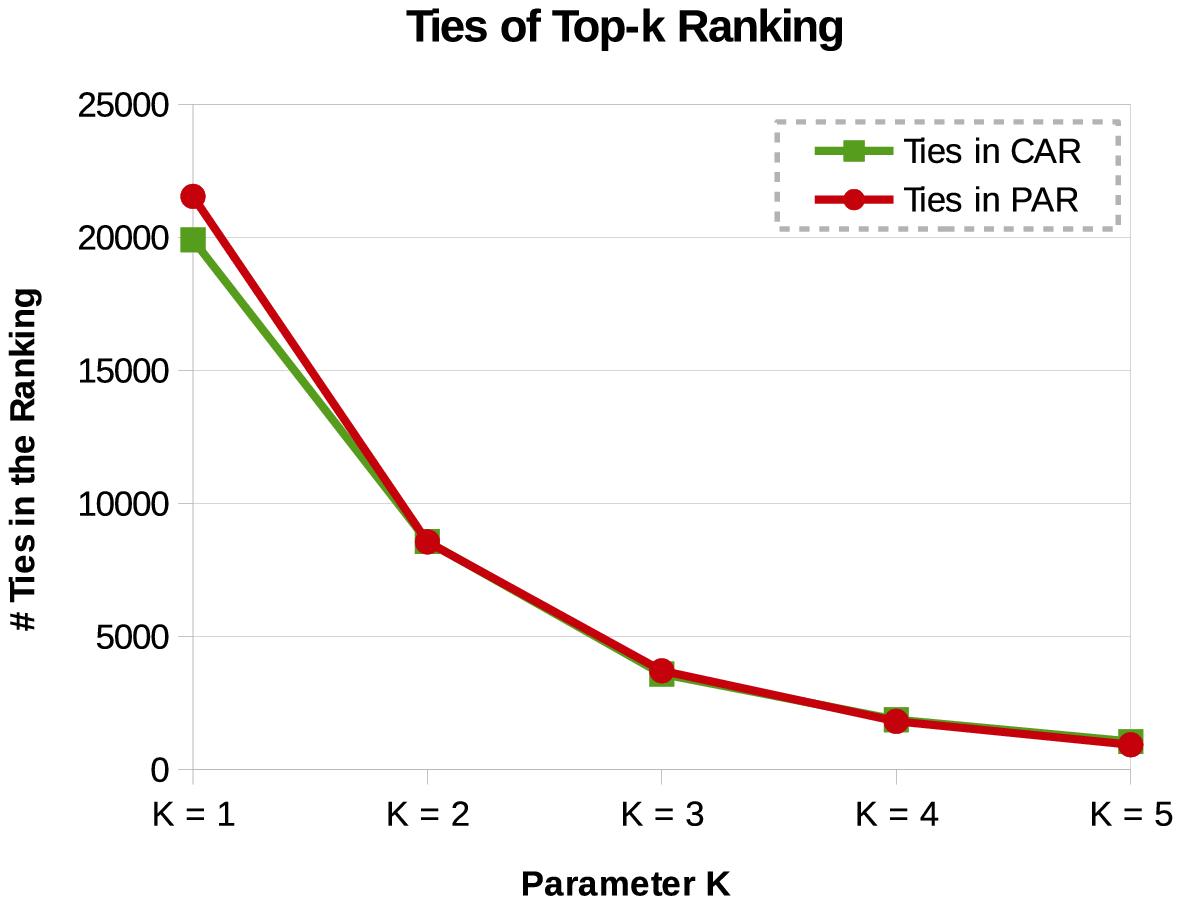,scale=0.6}
		\caption{Top-l Ranking}
		\label{fig:K_Ranking}
	\end{subfigure}
\caption{ Analysis of Parameter $k$  in NED}
\label{Fig:Parameter $k$ Analysis in NED}
\end{figure*}

In Figure \ref{fig:TED* Computation Time}, we plot the time to compute TED* with different tree sizes. In this experiment, we extract $3$-adjacent trees from nodes in (AMZN) graph and (DBLP) graphs respectively.  As shown in the Section \ref{sec:TED*, TED and GED Comparison}, the exact TED and GED cannot deal with trees and graphs with more than $10$ nodes. However, TED* is able to compute the distance between a pair of trees with up to $500$ nodes in one millisecond.

Figure \ref{fig:NED Computation Time} plots the NED computation time for different tree levels as $k$ changes. We randomly select $1000$ nodes from (CAR) graph and (PAR) graph. For each node, we extract the $k$-adjacent trees where $k$ varies from $1$ to $8$. For different $k$, the average NED computation time is computed for each pair. It is obvious that when the parameter $k$ increases, the computation time increases. When the value of $k$ is under $5$, the computation time is within one millisecond. Next, we show that, the parameter $k$ does not need to be very large ($5$ is large enough) for nearest neighbor queries and top-$l$ ranking queries to give meaningful results.

\subsection{Analysis of Parameter k}
\label{sec:Parameter K Analysis}

There is only one parameter $k$ in NED which is number of levels (hops) of neighbors to be considered in the comparison. In this section, we use nearest neighbor and top-$l$ ranking queries to show the effects of parameter $k$ on the query results.

The nearest neighbor query task is the following: for a given node in one graph, find the $l$ most similar nodes in another graph. When the parameter $k$ is small, more nodes in the graph can have the same minimal distance (usually $0$) to a given node.  When $k$ increases, the NED increases monotonically as proved in Section \ref{sec: Parameter K and Monotonicity}. Therefore, by choosing different parameter $k$, we can control the number of nodes in the nearest neighbor result set. Figure \ref{fig:K_NN} shows the number of nodes in the nearest neighbor result set for different values of  $k$. In the experiment, we randomly pick $100$ nodes from (CAR) and (PAR) graphs as queries and the nodes in the other graph are computed. It is obvious that when the parameter $k$ increases, the number of nodes in the nearest neighbor result set decreases.

The effect of parameter $k$ for the top-$l$ ranking query indicates how many identical distances (ties) that appear in the ranking. As shown in Figure \ref{fig:K_Ranking}, the ties start to break when $k$ increases. Intuitively, it is more likely to have isomorphic neighborhood structures if fewer levels of structures are considered. Figure \ref{fig:K_Ranking} shows the number of ties in the top-$l$ ranking for different values of $k$.  The experimental setting is the same as in the nearest neighbor query.

Choosing a proper value for the parameter $k$ depends on the query speed and quality. When $k$ increases, the computation time of NED increases as shown in Section \ref{sec:TED* and NED Computation}. On the other hand, when $k$ increases, both the number of nodes in the nearest neighbor result set and the number of ties in the ranking decreases. So it is clear that there exists a trade-off between the query speed and quality. Furthermore, the proper value of $k$ depends on the specific applications that the graphs come from.

\subsection{Query Comparison}
\label{sec:Query Comparison}

\begin{figure*}
\centering
	\begin{subfigure}[b]{0.48\textwidth}
		\epsfig{file=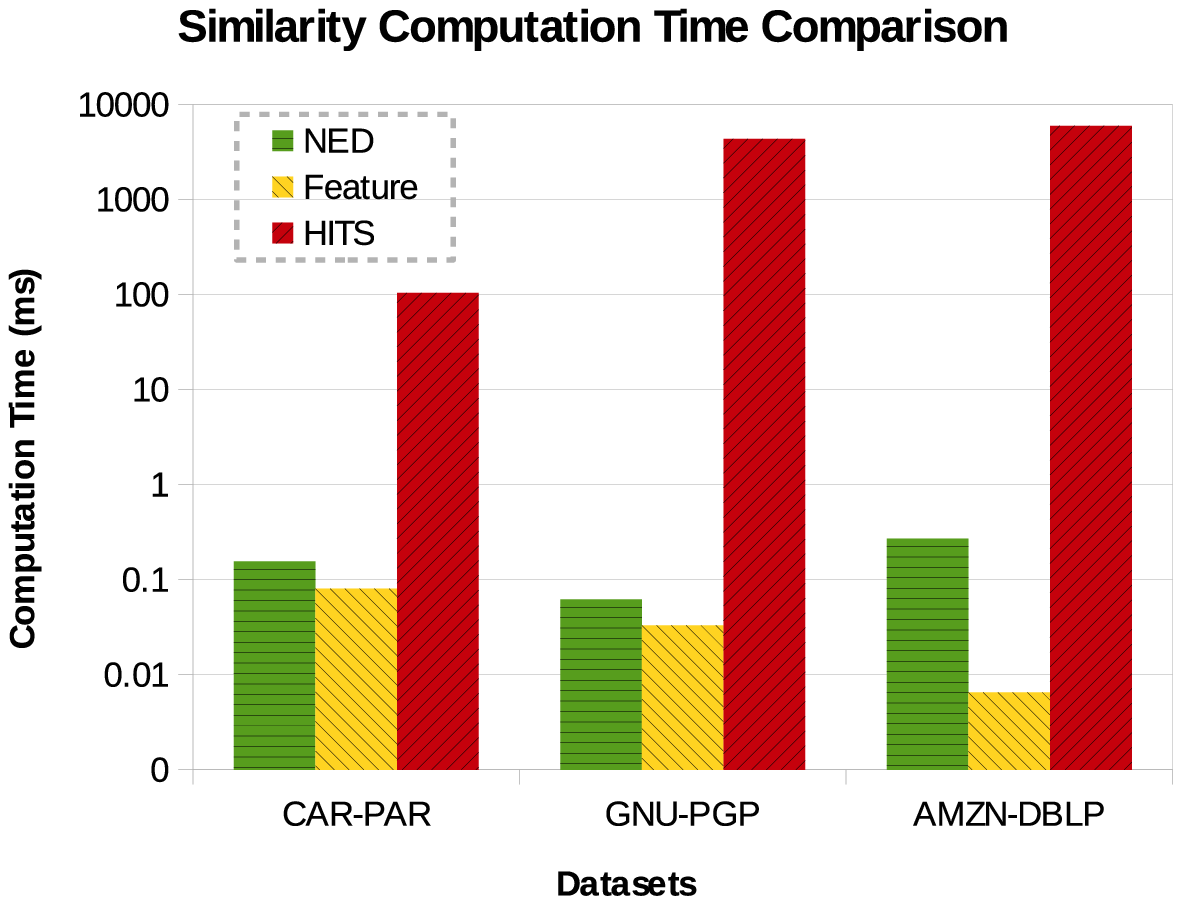,scale=0.6}
		\caption{Computation Time}
		\label{fig:NED_Feature_HITS_Time} 
	\end{subfigure}
	\begin{subfigure}[b]{0.48\textwidth}
		\epsfig{file=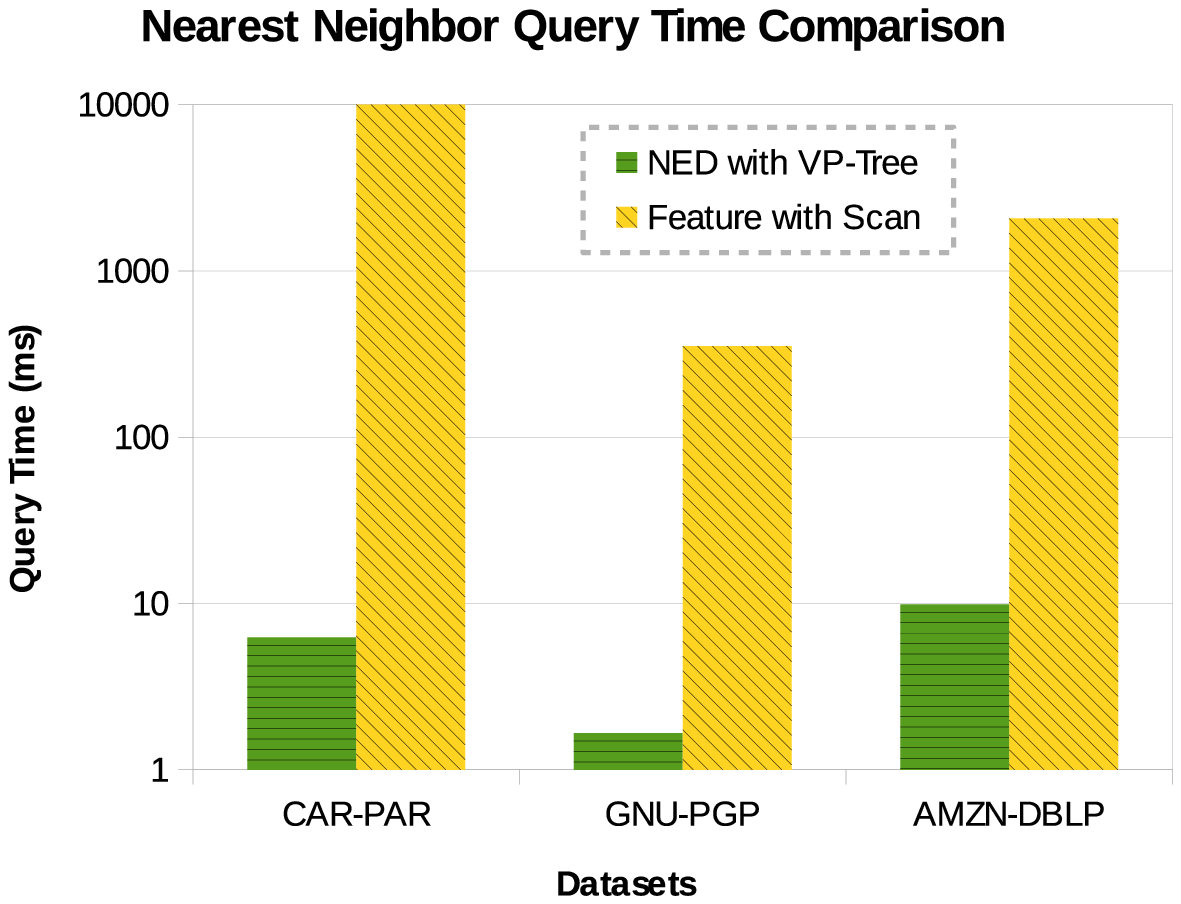,scale=0.6}
		\caption{NN Query Time}
		\label{fig:NED_Feature_HITS_Query} 
	\end{subfigure}
\caption{Node Similarity Comparison}
\label{Fig:Node Similarity Comparison}
\end{figure*}

In this section, we compare the performance of NED with other existing inter-graph node similarity measures: HITS-based similarity and Feature-based similarity.

Figure \ref{fig:NED_Feature_HITS_Time} shows the distance computation time for NED, HITS-based similarity, and Feature-based similarity. In this experiment, we extract $5$-adjacent trees for the nodes in (CAR) and (PAR) graphs and $3$-adjacent trees for the nodes in (PGP), (GNU), (AMZN) and (DBLP) graphs. NED, HITS-based similarity, and Feature-based similarity are computed over random pairs of nodes and the average computation time for different measures are shown in Figure \ref{fig:NED_Feature_HITS_Time}.

From this figure, it is clear that HITS-based similarity is the slowest among all three methods, because the HITS-based similarity iteratively updates the similarity matrix until the matrix converges. Feature-based similarity is faster than NED which makes sense since Feature-based similarity only collects statistical information from the neighborhood. NED pays a little extra overhead to take into account more topological information and be a metric. We show later why more topological information and metric matter.

As discussed, the Feature-based similarity discards certain topological information which makes it not precise. We use graph de-anonymization in Section \ref{sec:Case Study: De-anonymizing Graphs} to show that, with more topological information, NED can achieve a higher precision in  de-anonymization compared to the Feature-based similarity.

Also since the Feature-based similarity uses different features for different pairs, the similarity values of two pairs of nodes are not comparable. When using the Feature-based similarity for nearest neighbor queries, a full scan is necessary. However, as a metric, NED has the advantage in being used with existing metric indices for efficient query processing. Figure \ref{fig:NED_Feature_HITS_Query} shows that although NED pays a little bit more time than Feature-based similarity in distance computation, by combining with a metric index (existing implementation of the VP-Tree), NED is able to execute a nearest neighbor query much faster (orders of magnitude) than the Feature-based similarity.

\subsection{Case Study: De-anonymizing Graphs}
\label{sec:Case Study: De-anonymizing Graphs}

In this section, we use graph de-anonymization as a case study to show the merits of NED. Since NED captures more topological information, it is capable to capture more robust differences between structures. Therefore, NED can achieve much higher precision than the Feature-based similarity.

In this experiment, similar to~\cite{DBLP:conf/kdd/HendersonGLAETF11}, we split (PGP) and (DBLP) graphs into training data and testing data. The training data is the graph with identification, while the testing data is the anonymous graph. As stated in \cite{DBLP:journals/tist/FuZX15}, we choose three methods to anonymize the graphs for testing data: naive anonymization, sparsification, and perturbation. For each node in the anonymous graph, we try to find top-$l$ similar nodes in the training data. If the top-$l$ similar nodes include the original identification of the anonymous node, we say that it successfully de-anonymizes this node. Otherwise, the de-anonymization fails.

In Figure \ref{fig:Deanonymize_PGP} and Figure \ref{fig:Deanonymize_DBLP}, we show the precision of de-anonymization using NED and Feature-based similarity. In the experiment, the parameter $k$ is set to $3$ and we examine the top-$5$ similar nodes  (best $5$ matches) in (PGP) data and the top-$10$ similar nodes in (DBLP) data. The permutation ratio in (PGP) is $1\%$ and in (DBLP) is $5\%$. Based on the results, NED is able to identify anonymized nodes with better accuracy than the Feature-based similarity.

In Figure \ref{fig:Permutation Ratio}, we show how the precision changes by varying permutation ratio. The precision of NED reduces slower than Feature-based similarity when the permutation ratio increases. The Figure \ref{fig:Top-K Finding} shows that when more nodes in top-$l$ results are examined, how the de-anonymization precision changes. It is clear that if fewer nodes are checked which means there are less nodes in the top-$l$ results, NED can still achieve a higher precision.

\begin{figure*}
\centering
	\begin{subfigure}[b]{0.48\textwidth}
		\epsfig{file=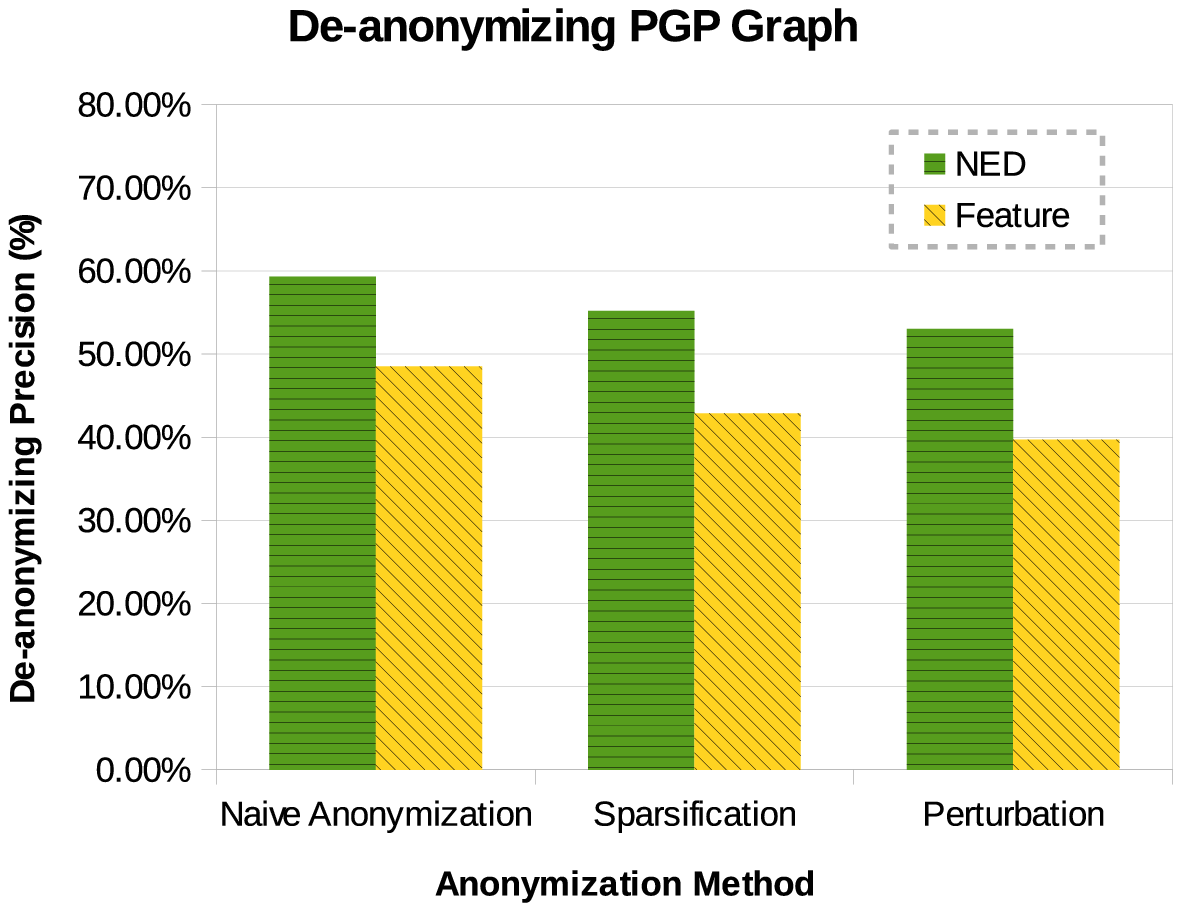,scale=0.6}
		\caption{De-Anonymize PGP}
		\label{fig:Deanonymize_PGP}
	\end{subfigure}
	\begin{subfigure}[b]{0.48\textwidth}
		\epsfig{file=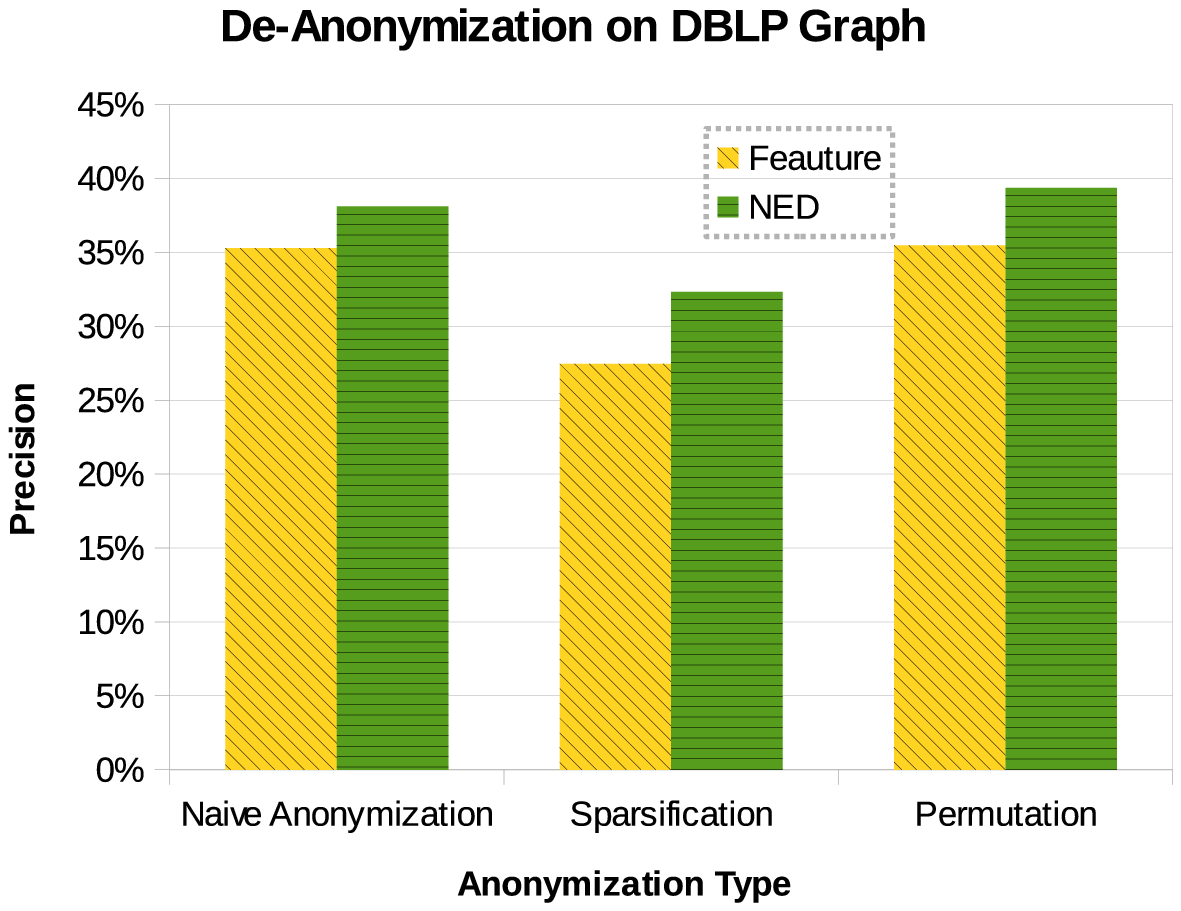,scale=0.6}
		\caption{De-Anonymize DBLP}
		\label{fig:Deanonymize_DBLP}
	\end{subfigure}
\caption{Graph De-Anonymization}
\label{Fig:Graph De-Anonymization}
\end{figure*}

\begin{figure*}
\centering
	\begin{subfigure}[b]{0.48\textwidth}
		\epsfig{file=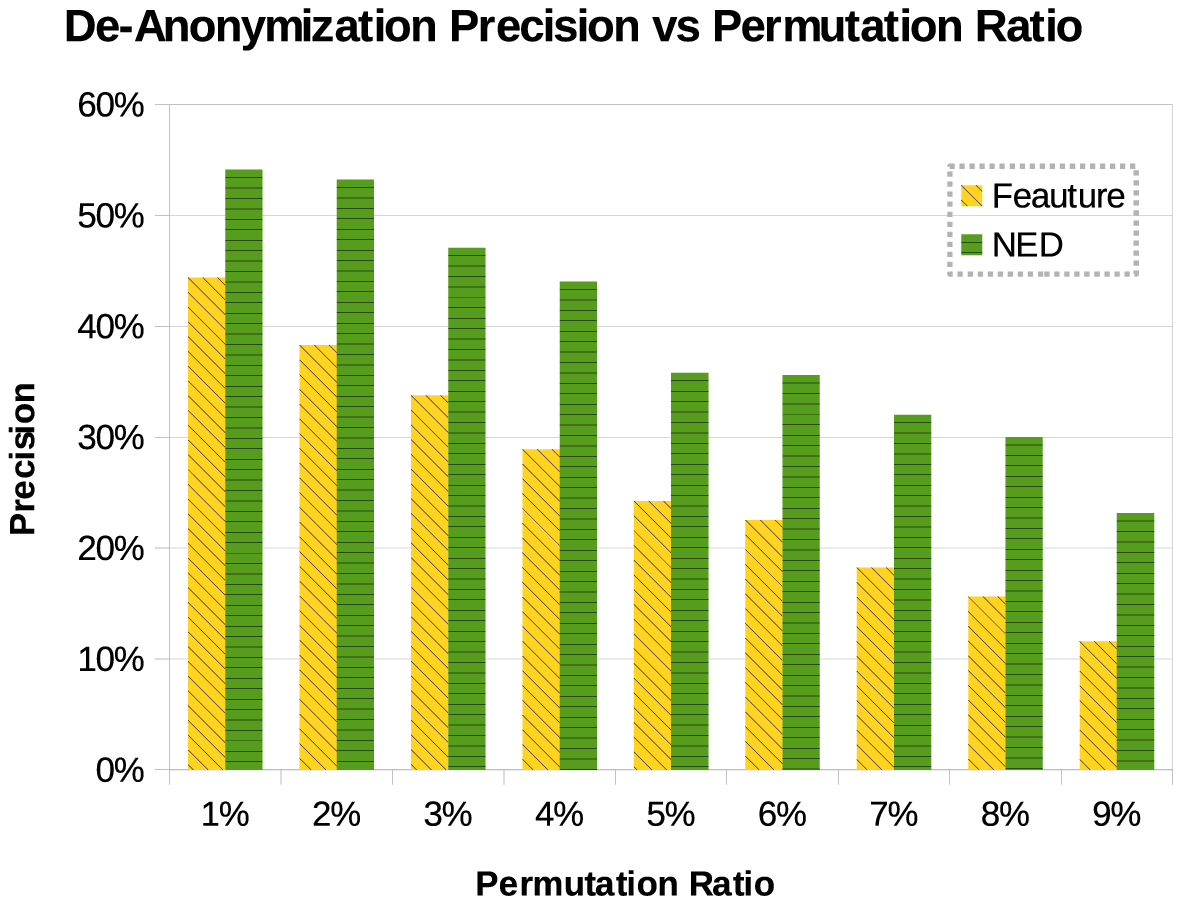,scale=0.6}
		\caption{Permutation Ratio}
		\label{fig:Permutation Ratio}
	\end{subfigure}
	\begin{subfigure}[b]{0.48\textwidth}
		\epsfig{file=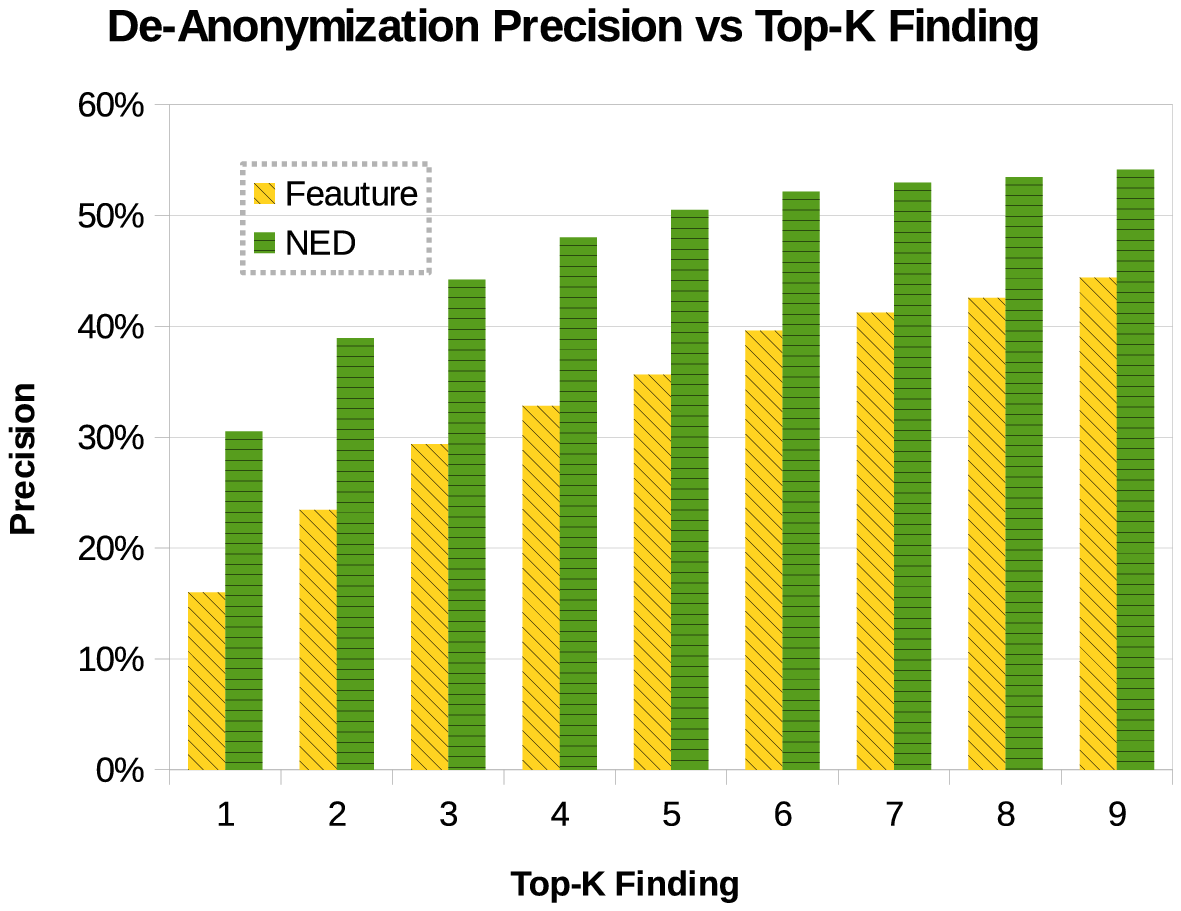,scale=0.6}
		\caption{Top-$l$ Finding}
		\label{fig:Top-K Finding}
	\end{subfigure}
\caption{Graph De-Anonymization Cont.}
\label{Fig:Graph De-Anonymization Cont.}
\end{figure*}

\section{Conclusion}
\label{Conclusion}
In this paper, we study the inter-graph node similarity problem. A major application of inter-graph node similarity is transfer learning on graphs, which means learning a new graph based on the existing knowledge from another one. To address this problem, this paper proposes a novel distance function called NED. In NED, the $k$-adjacent trees between two nodes to be compared are extracted first. Then, a newly proposed modified tree edit distance called TED* is used calculate the distance between two $k$-adjacent trees. TED* is a generic distance function for comparing trees which is easy to compute in polynomial time and satisfies the metric properties of the edit distance. 
Therefore, as a measurement on nodes, NED is a node metric. Due to the metric properties, NED is compatible with existing metric indexing methods.  Moreover, since NED captures more structural information, it is demonstrated to be a more effective and precise for graph de-anonymization than other existing methods. Indeed, we empirically verify the efficiency and effectiveness of NED using real-world graphs.

\bibliographystyle{abbrv}
\bibliography{zhu}

\begin{thebibliography}{10}

\bibitem{DBLP:conf/pakdd/AkogluMF10}
L.~Akoglu, M.~McGlohon, and C.~Faloutsos.
\newblock oddball: Spotting anomalies in weighted graphs.
\newblock In {\em PAKDD}, pages 410--421, 2010.

\bibitem{DBLP:journals/pvldb/AntonellisGC08}
I.~Antonellis, H.~Garcia{-}Molina, and C.~Chang.
\newblock Simrank++: query rewriting through link analysis of the click graph.
\newblock {\em PVLDB}, 1(1):408--421, 2008.

\bibitem{DBLP:conf/asunam/BerlingerioKEF13}
M.~Berlingerio, D.~Koutra, T.~Eliassi{-}Rad, and C.~Faloutsos.
\newblock Network similarity via multiple social theories.
\newblock In {\em ASONAM}, pages 1439--1440, 2013.

\bibitem{Blondel2004}
V.~D. Blondel, A.~Gajardo, M.~Heymans, P.~Senellart, and P.~V. Dooren.
\newblock A measure of similarity between graph vertices: Applications to
  synonym extraction and web searching.
\newblock {\em SIAM Rev.}, 46(4):647--666, 2004.

\bibitem{DBLP:journals/bioinformatics/ClarkK14}
C.~Clark and J.~Kalita.
\newblock A comparison of algorithms for the pairwise alignment of biological
  networks.
\newblock {\em Bioinformatics}, 30(16):2351--2359, 2014.

\bibitem{DBLP:journals/bioinformatics/DavisYMSP15}
D.~Davis, {\"{O}}.~N. Yaveroglu, N.~Malod{-}Dognin, A.~Stojmirovic, and
  N.~Przulj.
\newblock Topology-function conservation in protein-protein interaction
  networks.
\newblock {\em Bioinformatics}, 31(10):1632--1639, 2015.

\bibitem{DBLP:journals/tist/FuZX15}
H.~Fu, A.~Zhang, and X.~Xie.
\newblock Effective social graph deanonymization based on graph structure and
  descriptive information.
\newblock {\em {ACM} {TIST}}, 6(4):49, 2015.

\bibitem{DBLP:journals/tssc/HartNR68}
P.~E. Hart, N.~J. Nilsson, and B.~Raphael.
\newblock A formal basis for the heuristic determination of minimum cost paths.
\newblock {\em {IEEE} Trans. Systems Science and Cybernetics}, 4(2):100--107,
  1968.

\bibitem{DBLP:conf/kdd/HendersonGLAETF11}
K.~Henderson, B.~Gallagher, L.~Li, L.~Akoglu, T.~Eliassi{-}Rad, H.~Tong, and
  C.~Faloutsos.
\newblock It's who you know: graph mining using recursive structural features.
\newblock In {\em KDD}, pages 663--671, 2011.

\bibitem{DBLP:conf/kdd/JehW02}
G.~Jeh and J.~Widom.
\newblock Simrank: a measure of structural-context similarity.
\newblock In {\em KDD}, pages 538--543, 2002.

\bibitem{DBLP:conf/kdd/JinLH11}
R.~Jin, V.~E. Lee, and H.~Hong.
\newblock Axiomatic ranking of network role similarity.
\newblock In {\em KDD}, pages 922--930, 2011.

\bibitem{DBLP:conf/sigmod/KhanLYGCT11}
A.~Khan, N.~Li, X.~Yan, Z.~Guan, S.~Chakraborty, and S.~Tao.
\newblock Neighborhood based fast graph search in large networks.
\newblock In {\em SIGMOD}, pages 901--912, 2011.

\bibitem{DBLP:journals/pvldb/KhanWAY13}
A.~Khan, Y.~Wu, C.~C. Aggarwal, and X.~Yan.
\newblock Nema: Fast graph search with label similarity.
\newblock {\em PVLDB}, 6(3):181--192, 2013.

\bibitem{DBLP:conf/www/Kunegis13}
J.~Kunegis.
\newblock {KONECT:} the koblenz network collection.
\newblock In {\em WWW, Companion Volume}, pages 1343--1350, 2013.

\bibitem{snapnets}
J.~Leskovec and A.~Krevl.
\newblock {SNAP Datasets}: {Stanford} large network dataset collection, June
  2014.

\bibitem{DBLP:journals/tkde/LiuHZLD13}
H.~Liu, J.~He, D.~Zhu, C.~X. Ling, and X.~Du.
\newblock Measuring similarity based on link information: {A} comparative
  study.
\newblock {\em {IEEE} Trans. Knowl. Data Eng.}, 25(12):2823--2840, 2013.

\bibitem{doi1971}
F.~Lorrain and H.~C. White.
\newblock Structural equivalence of individuals in social networks.
\newblock {\em J. Math. Sociol.}, 1(1):49--80, 1971.

\bibitem{DBLP:journals/bioinformatics/Malod-DogninP15}
N.~Malod{-}Dognin and N.~Przulj.
\newblock {L-GRAAL:} lagrangian graphlet-based network aligner.
\newblock {\em Bioinformatics}, 31(13):2182--2189, 2015.

\bibitem{DBLP:conf/sspr/NeuhausRB06}
M.~Neuhaus, K.~Riesen, and H.~Bunke.
\newblock Fast suboptimal algorithms for the computation of graph edit
  distance.
\newblock In {\em SSPR}, pages 163--172, 2006.

\bibitem{DBLP:journals/pvldb/PawlikA11}
M.~Pawlik and N.~Augsten.
\newblock Rted: A robust algorithm for the tree edit distance.
\newblock {\em PVLDB}, 5(4):334--345, 2011.

\bibitem{DBLP:journals/bioinformatics/PrzuljCJ04}
N.~Przulj, D.~G. Corneil, and I.~Jurisica.
\newblock Modeling interactome: scale-free or geometric?
\newblock {\em Bioinformatics}, 20(18):3508--3515, 2004.

\bibitem{ASI1973}
H.~Small.
\newblock Co-citation in the scientific literature: A new measure of the
  relationship between two documents.
\newblock {\em J. Assoc. Inf. Sci. Technol.}, 24(4):265--269, 1973.

\bibitem{DBLP:journals/pvldb/SunHYYW11}
Y.~Sun, J.~Han, X.~Yan, P.~S. Yu, and T.~Wu.
\newblock Pathsim: Meta path-based top-k similarity search in heterogeneous
  information networks.
\newblock {\em PVLDB}, 4(11):992--1003, 2011.

\bibitem{DBLP:journals/jacm/Tai79}
K.~Tai.
\newblock The tree-to-tree correction problem.
\newblock {\em J. {ACM}}, 26(3):422--433, 1979.

\bibitem{DBLP:conf/icdm/TongFP06}
H.~Tong, C.~Faloutsos, and J.~Pan.
\newblock Fast random walk with restart and its applications.
\newblock In {\em ICDM}, pages 613--622, 2006.

\bibitem{DBLP:journals/tkde/WangWYY12}
G.~Wang, B.~Wang, X.~Yang, and G.~Yu.
\newblock Efficiently indexing large sparse graphs for similarity search.
\newblock {\em {IEEE} Trans. Knowl. Data Eng.}, 24(3):440--451, 2012.

\bibitem{DBLP:conf/kdd/XuYFS07}
X.~Xu, N.~Yuruk, Z.~Feng, and T.~A.~J. Schweiger.
\newblock {SCAN:} a structural clustering algorithm for networks.
\newblock In {\em KDD}, page 8240833, 2007.

\bibitem{DBLP:journals/pvldb/YuLZCP13}
W.~Yu, X.~Lin, W.~Zhang, L.~Chang, and J.~Pei.
\newblock More is simpler: Effectively and efficiently assessing node-pair
  similarities based on hyperlinks.
\newblock {\em PVLDB}, 7(1):13--24, 2013.

\bibitem{DBLP:journals/pvldb/ZengTWFZ09}
Z.~Zeng, A.~K.~H. Tung, J.~Wang, J.~Feng, and L.~Zhou.
\newblock Comparing stars: On approximating graph edit distance.
\newblock {\em PVLDB}, 2(1):25--36, 2009.

\bibitem{DBLP:journals/ipl/ZhangJ94}
K.~Zhang and T.~Jiang.
\newblock Some {MAX} snp-hard results concerning unordered labeled trees.
\newblock {\em Inf. Process. Lett.}, 49(5):249--254, 1994.

\bibitem{DBLP:journals/ipl/ZhangSS92}
K.~Zhang, R.~Statman, and D.~Shasha.
\newblock On the editing distance between unordered labeled trees.
\newblock {\em Inf. Process. Lett.}, 42(3):133--139, 1992.

\end{thebibliography}

\appendix

\section{NED for Graph Similarity}
\label{sec: NED for Graph Similarity}

At last, we propose another promising application for node metric as one of our future work. Since a graph can be seen as a collection of nodes, one popular direction for graph similarity is to compare pairwise nodes from two graphs. If there are more paris of nodes from two graphs can be matched, the two graphs are more similar. 

There are many existing distance functions can be used to measure the similarity between two collections such as: Hausdorff distance, Fr\'{e}chet distance, Earth mover's distance and etc. However, to measure two collections of nodes, the underline distance between pairwise nodes should be a metric. Fortunately, NED proposed in our paper is a metric for inter-graph nodes.

Next, we use Hausdorff distance as an example to show how to calculate the distance between two graphs based on NED.

\begin{Definition}
Given two graphs $A = \{ a_1, a_2, ..., a_p \}$ and $B = \{ b_1, b_2, ..., b_q \}$. The Hausdorff distance between two graphs is defined as:
\begin{equation}
H(A, B) = max(h(A, B), h(B, A))
\end{equation}
where
\begin{equation}
h(A, B) = \max_{a \in A}{ \min_{b \in B} (\delta_T(T(a, k), (b, k)) )}
\end{equation}
\label{Definition:Graph_Matching}
\end{Definition}

In Definition \ref{Definition:Graph_Matching}, $a_i$ and $b_j$ are nodes in graphs $A$ and $B$ respectively. While $\delta_T(T(a, k), (b, k))$ is the TED* distance between the $k$-adjacent trees of node $a$ and $b$. When the distance function for the corresponding data points is metric as NED, the Hausdorff distance is also metric. Therefore, we provide a metric distance function to measure the similarity between two graphs. Since both Hausdorff distance and NED are polynomial-time computable. This metric for measuring graph similarity is also polynomial-time computable.

\end{document}